%% file: main_2_.tex
\documentclass[a4paper,UKenglish,cleveref, autoref, thm-restate]{esa2}
\usepackage{layout}
\usepackage[table, dvipsnames]{xcolor}
\usepackage{graphicx} 
\usepackage{enumerate}
\usepackage{amsmath,mathtools}
\usepackage{amsthm}
\usepackage{array} 
\usepackage{algpseudocode}
\usepackage{algorithm}
\usepackage{amsfonts}
\usepackage{tabularray}

\usepackage{tikz}
\usepackage{nicefrac}
\usetikzlibrary{patterns, patterns.meta}
\newtheorem{defn}[theorem]{Definition}

\newtheorem{prop}[theorem]{Proposition}
\newtheorem{obs}[theorem]{Observation}
\newtheorem{inv}[theorem]{Invariant}
\newtheorem{cor}[theorem]{Corollary}
\newtheorem{asp}[theorem]{Assumption}

\newtheorem{rul}[theorem]{Rule}

\DeclareMathOperator{\onl}{\mathrm{onl}}
\DeclareMathOperator{\offl}{\mathrm{offl}}
\DeclareMathOperator{\offlinecolors}{\mathrm{blue}, \mathrm{red}, \mathrm{yellow}}
\newcommand{\Vast}{\bBigg@{5}}

\makeatother

\tikzset{myblue/.style={fill=white, draw=black,  circle, pattern=horizontal lines light blue, minimum size=6mm, inner sep=.5mm}}
\tikzset{myyellow/.style={draw=black, fill=Goldenrod, circle,  minimum size=6mm,inner sep=.5mm}}
\tikzset{myred/.style={draw=black,  circle, pattern=mycheck, pattern color=black, minimum size=6mm,inner sep=.5mm}}
\tikzset{neutral/.style={draw=black,  circle, fill=white, minimum size=6mm,inner sep=.5mm}}

\input{checkerboard}
\title{Online Makespan Scheduling under Scenarios} 

\author{Ekin Ergen}{Technical University of Berlin, Germany \and \url{http://www.myhomepage.edu} }{ergen@math.tu-berlin.de}{https://orcid.org/0000-0002-1825-0097}{}

\authorrunning{E. Ergen} 

\Copyright{Ekin Ergen} 

\ccsdesc[100]{Theory of computation~Online algorithms} 

\keywords{online scheduling, scenario-based model, online algorithms} 

\category{} 

\relatedversion{} 

\funding{The author was supported by the Deutsche Forschungsgemeinschaft (DFG, German Research Foundation) under Germany's Excellence Strategy --- The Berlin Mathematics Research Center MATH+ (EXC-2046/1, project ID: 390685689)}

\acknowledgements{I would like to thank Martin Skutella for inspiring me to consider this problem, for his invaluable suggestions and for the insightful discussions. I am also grateful to the anonymous reviewers for their careful reading and valuable feedback.}

\EventEditors{Anne Benoit, Haim Kaplan, Sebastian Wild, and Grzegorz Herman}
\EventNoEds{4}
\EventLongTitle{33rd Annual European Symposium on Algorithms (ESA 2025)}
\EventShortTitle{ESA 2025}
\EventAcronym{ESA}
\EventYear{2025}
\EventDate{September 15--17, 2025}
\EventLocation{Warsaw, Poland}
\EventLogo{}
\SeriesVolume{351}
\ArticleNo{25}

\begin{document}
\maketitle
\begin{abstract}
We consider a natural extension of online makespan scheduling on identical parallel machines by introducing scenarios. A scenario is a subset of jobs, and the task of our problem is to find a global assignment of the jobs to machines so that the maximum makespan under a scenario, i.e., the maximum makespan of any schedule restricted to a scenario, is minimized. 

For varying values of the number of scenarios and machines, we explore the competitiveness of online algorithms. We prove tight and near-tight bounds, several of which are achieved through novel constructions. In particular, we leverage the interplay between the unit processing time case of our problem and the hypergraph coloring problem both ways: We use hypergraph coloring techniques to steer an adversarial family of instances proving lower bounds, which in turn leads to lower bounds for several variants of online hypergraph coloring.
\end{abstract}

\newpage
\section{Introduction}
We study a natural extension of \emph{online makespan scheduling}, also known as \emph{online load balancing}, one of the most well-known problems in the scope of online optimization. In online makespan scheduling, we know a number $m$ of identical parallel machines in advance, in addition to which jobs $j$ are revealed one at a time along with their processing times $p_j\geq 0$. Our task is to assign the jobs to a machine as soon as they are revealed, solely under the knowledge of the preceding jobs, so as to minimize the \emph{makespan}, i.e., the maximum sum of processing times any machine is assigned to. 
This problem was initially studied by Graham almost 60 years ago, who proved that assigning every job to a currently least loaded machine is a $(2-\frac{1}{m})$-competitive algorithm \cite{graham}. It was observed that no better deterministic algorithm can exist for $m\in \{2,3\}$ \cite{graham}\cite{grahamlb}, while for larger values of $m$, such a tight ratio is still not known. 

We extend online makespan scheduling by a \emph{scenario-based model}: We are given a number $m$ of machines. There is a known number $K$ of job sets (called \emph{scenario}s) $S_k\subseteq \{1,\ldots,n\}$, $k\in [K]$, over a common ground set $\{1,\ldots, n\}$. Although the number of scenarios is known, neither the number of jobs nor their processing times $p_j$ ($j\in [n]$) are known in advance. The ground set $[n]$ is revealed in increasing order of jobs; more precisely, we are incrementally informed of $p_1,\ldots, p_j$ as well as $S_k\cap \{1,\ldots, j\}$ for increasing $j$.
The task is to find a global assignment $\tau \colon \{1,\ldots, n\}\to\{1,\ldots, m\}$ of the jobs to machines that delivers a reasonable solution restricted to every scenario. To be more precise, the objective is to minimize the makespan of the \emph{worst-case scenario}, i.e., the scenario that yields the highest makespan. As in the classical online scheduling problem, a job $j$ must be assigned irrevocably to a machine $\tau(j)$ as soon as it is revealed. In essence, we are solving a finite number of instances of Online Makespan Scheduling over the same ground set of jobs simultaneously.

The simplicity of the scenario model offers a variety of interpretations and applications. For instance, we can view the scenarios as distinct services that are provided in $m$ (number of machines) sites to $n$ (number of jobs) customers, where the customers are revealed one by one with their desired services and need to be assigned to one of the sites immediately.
Another interpretation of our problem may concern a fair balancing of goods according to multiple agents in the following setting: Each of the $K$ agents are interested in partitioning a subset $S_k$, $k\in \{1,\ldots, K\}$ of $n$ goods with weights $p_j$ among $m$ buckets, where the partition should be as balanced for all agents as possible (i.e., we would like to minimize the maximum bucket load for any agent). Here, the goods are revealed one by one together with the agents that are interested in them.

Moreover, even special cases of online makespan scheduling under scenarios correspond to interesting questions in their own regard. The special case $p_j\equiv 1$ of unit processing times can be related to some cases of discrepancy minimization as well as hypergraph coloring.  

In the large array of previous research on combinatorial optimization under scenarios, it was shown for several offline problems that optimizing under a number of scenarios is computationally harder than their underlying single-scenario counterparts. Accordingly, our goal in this work is to gain a better understanding of the intersection of scenario-based models and competitiveness, in particular how competitiveness becomes harder under a growing number of scenarios.
\subsection{Our Contribution}\label{contribution}

Implementing a miscellany of ideas to analyze competitiveness bounds for the varying numbers $m$ of machines and $K$ of scenarios, we achieve several tight bounds and further nontrivial bounds. An overview of our results, depending on the numbers $m$ and $K$, can be seen in Table \ref{tab:results}. Several of the entries are obtained via a generalization of Graham's List Scheduling Algorithm which we extend using the pigeonhole principle. 

By a simple example, we show that for any number $m\geq 2$ of machines and (at least) $3$ scenarios, no algorithm with a competitive ratio better than $2$ can exist, even for jobs with unit processing times. To underline the competitiveness gap between two and three scenarios, we contrast the aforementioned result with an algorithm for two machines and two scenarios, which is one of our main results.

\begin{theorem}\label{thm:53ub}
    There exists a $\nicefrac{5}{3}$-competitive algorithm for \textsc{Online Makespan Scheduling under Scenarios} for $m=K=2$.
\end{theorem}
This algorithm as well as its analysis combine techniques  well-known from algorithms for online makespan scheduling with novel ideas that concern the scenario-based properties of a schedule. 
What sets our analysis apart from previous work is that we alternate between forward and reverse engineering: Assuming that our algorithm obeys some fixed rules, we explore what we may assume about a well-structured worst-case instance. 
In spite of seemingly strong restrictions that we impose, the ratio of our algorithm is not far from the best achievable bound by any deterministic algorithm, which we show is at least $\frac{9+\sqrt{17}}{8}\approx 1.640$. Altogether, we obtain a near-complete picture of competitiveness for $m=2$.

The remainder of the paper demonstrates a disparity between increasing the number $m$ of machines and increasing the number $K$ of scenarios. 
For a fixed but arbitrary $m$, we show that there exists a sufficiently large number $K(m)$ of scenarios and a family of instances with $K(m)$ scenarios on which no deterministic online algorithm can achieve a ratio better than $m$. Furthermore, we can even choose all instances to have unit processing times $p_j\equiv 1$. This result is particularly noteworthy, considering that any online algorithm on $m$ machines is trivially $m$-competitive: Indeed, the online output is at most the total load and an offline optimum is at least the average load, and these are by at most a factor of $m$ apart. 

To find the suitable number $K(m)$, we interpret the instances as hypergraphs, passing to a problem we call \textsc{Online Makespan Hypergraph Coloring$(m)$} instead (cf.~Section \ref{sec:prelim}). Strikingly, such a family can be created using hypertrees on which no deterministic algorithm can achieve a coloring without a monochromatic edge simultaneously. 

\begin{theorem}\label{thm:general}
   Let $m\in \mathbb N$. There exists a number $K(m)$ such that for any $r<m$, no deterministic algorithm for \textsc{Online Makespan Hypergraph Coloring$(m)$} is $r$-competitive, even when restricted to hyperforests with $K(m)$ hyperedges.  
\end{theorem}

The statement of Theorem \ref{thm:general} is best possible in the sense that the numbers $K(m)$ cannot be replaced by a global $K$ for which the incompetitiveness holds for all $m$ (cf.~Corollary \ref{cor:k1competitive}). 

The proof of Theorem \ref{thm:general} can be found in the full version of the paper. To showcase some of our techniques, we sketch a proof for a similar statement for $K=3$ in Section \ref{sec:hypergraphs}. We utilize a subhypergraph on $7$ nodes as a gadget and obtain a relatively compact construction, both in terms of the hypergraph and the case distinction that we apply to analyze it. 

An advantage of our approach while proving Theorem \ref{thm:general} is that it only features hyperforests. More precisely, we implicitly show that for every $m\in \mathbb N$, there is a number $n\leq (2m+1)\uparrow\uparrow5$ of nodes such that no deterministic online algorithm on $n$ nodes that uses $m$ colors can avoid a monochromatic hyperedge. In other words, in order to color a hyperforest with an arbitrary number $n$ of nodes online without monochromatic edges, we need $\Omega(\mathrm{slog_5} n)$ colors (cf.~Preliminaries for notation), even if the restrictions $e\cap \{1,\ldots, j\}$ of hyperedges $e$ to revealed nodes are known upon the revelation of nodes $\{1,\ldots, j\}$. Here, the emphasis is on the additional knowledge of partial hyperedges: This makes an adversarial revelation of hyperedges more challenging and to the best of our knowledge, none of the previous lower bound constructions \cite{NAGYGYORGY200823}\cite{hypertreecoloring} are applicable under such additional information. 
This result addresses a completely open problem from \cite{NAGYGYORGY200823} and sets a non-constant lower bound for Online Hypergraph Coloring for hypertrees under the knowledge of partial hyperedges.

\begin{cor}\label{ackermann}
    Let $n\in \mathbb N$. There exists no $o(\mathrm{slog_5} n)$-competitive deterministic algorithm for Online Hypergraph Coloring with nodes $v_1,\ldots, v_n$ (revealed in this order) even if we restrict to  hypertrees and additionally, the subsets $e\cap \{v_1,\ldots, v_j\}$ of hyperedges $e\in E$ are known upon the assignment of node $v_j$.
\end{cor}

For unit processing times, $K=3$ and varying $m$, we further propose a $2$-competitive algorithm, which is best possible due to the lower bounds mentioned earlier. Although the special case of unit processing times is admittedly a restriction, the result still illustrates a contrast with the tight lower bound of $m$ for a varying number of scenarios, which also holds already for instances with unit processing times.

\begin{table}[]

\begin{center}
\centering
 \begin{tblr}{width=\paperwidth,
      colspec={|c|ccc|},
      row{1}={font=\bfseries},
      row{1}={bg=orange!30},
      column{1}={font=\bfseries},
      row{even}={bg=orange!10},
    }

\hline
  m $\downarrow$ K $\rightarrow$ & 2 & 3 & 4+\\ 
  \hline
  2 & $\in \left(1.640, 
{{\frac{5}{3}}}
  \right]$ &{${=2}$}
  &{${=2}$}\\ 
  3 & ${\in\left[\frac{5}{3},2\right]}$& 
  \begin{tabular}{@{}c@{}}
 $[2,3]$  
\\ \footnotesize {${=2}$ for ${p_j\equiv 1}$}
\end{tabular} 
 &
 \begin{tabular}{@{}c@{}}
{$\in[2,3]$}
  \\\footnotesize {$=3$ even for $p_j\equiv 1$,} \\ \footnotesize{if $K\geq 233$ }
\end{tabular}
\\ 
   4+ & $\in \left[\frac{5}{3}, 3\right)$ &  \begin{tabular}{@{}c@{}}
 {$\in[2,4)$}\\ 
 {\textcolor{black}{\footnotesize {${=2}$ for ${p_j\equiv 1}$}}}\end{tabular}
   &\begin{tabular}{@{}c@{}}\small
 {$\in[2,\min\{m,K+1\})$}
   \\{\footnotesize $=m$ even for $p_j\equiv 1$,} \\\footnotesize{ if $K$ large enough}
 \end{tabular}
 \\
 \hline
  \end{tblr}

\end{center}
    \caption{Overview of the competitive ratios obtained by our results. Recall that on $m$ machines, any algorithm is trivially $m$-competitive.}
    \label{tab:results}
\end{table}
\subsection{Preliminaries}\label{sec:prelim}
Throughout this paper, for $n\in \mathbb N$, we denote by $[n]\coloneqq\{1,\ldots,n\}$ the set of numbers $1$ through $n$. An instance $ I=(n, (p_j)_{j\in [n]}, \{S_k\}_{k\in [K]})$ of \textsc{Online Makespan Scheduling under Scenarios$(m,K)$} (\textsc{OMSS$(m,K)$}) is parametrized by a number $m\in \mathbb N$ of machines and a number $K\in \mathbb N$ of scenarios, and defined as a tuple consisting of a number $n\in \mathbb N$ of jobs, a processing time $p_j\geq 0$ for each job $j$ and job subsets $S_1,\ldots, S_K\subseteq [n]$, called \emph{scenarios}. As we usually work with a single instance at a time, we often suppress this notation in the sequel as well as the parameters $m$ and $K$, which are often specified per section.

For a set $S\subseteq [n]$, we denote $p(S)\coloneqq \sum_{j\in S}p_j$. Our objective is to find a partitioning $J_1\dot \cup \ldots \dot\cup J_m$ of the set of jobs that minimizes 
$\max_{k\in [K]}\max_{i\in [m]}p(J_i\cap S_k).$
In other words, we consider schedules that arise from restricting a partition to each of the scenarios, and our goal is then to minimize the makespan of the worst-case scenario, which we also call the \emph{makespan of the schedule}. Furthermore, our instance is \emph{online} in the sense that the number $m$ of machines is known in advance, but jobs $j\in [n]$ as well as their processing times $p_j$ and the indices $k\in [K]$ such that $j\in S_k$ are revealed one job $j$ at a time. Once the information of a job $j$ is revealed in addition to already known information about the jobs in $[j-1]$, the job $j$ must be assigned to a machine $i\in [m]$, which is an irrevocable decision.

In our analyses and arguments, we use the notion of a \emph{completion time} of a job $j$. In classical scheduling, this is defined as the sum of processing times preceding (including) $j$ that are on the same machine as $j$. In our setting, the restricted schedules are evaluated without idle times, meaning that such completion times might be different for different scenarios. Accordingly, we define the \emph{completion time of job $j\in J_i\cap S_k$ in scenario $S_k$} as $C^k_j\coloneqq \sum_{j'\in J_i\cap S_k\cap [j]}p_{j'}$ and the \emph{completion time} as $C_j\coloneqq \max_{k\in [K], j\in S_k}C_j^k$. We can further define 
\[\mathrm{makespan}(j)\coloneqq \max_{j'\in [j]}C_{j'}.\]

Online algorithms are often analyzed by their performance compared to the best possible outcome under full information. In existing literature, this comparison is quantified by the \emph{competitive ratio}. An algorithm is said to have competitive ratio $\rho$ if for every instance, it is guaranteed to output a solution that has an objective value which deviates from the value of an optimal solution under full information by at most a factor of $\rho$.

We also consider the special case of unit processing times $p_j\equiv 1$. In this special case, we may view our jobs as nodes, our scenarios as hyperedges containing a subset of the nodes and the partition $J_1\dot\cup \ldots \dot \cup J_m$ as a coloring $c\colon [n]\to[m]$ of the nodes. Then, our problem can be formulated as an online hypergraph coloring problem, in which we color a hypergraph $\mathcal H=(V,E)$ with a fixed number of colors $1,\ldots, m$ and seek to minimize $\max_{i\in [m],e\in E}|\{c^{-1}(i)\}\cap e|$.
The nodes $j\in [n]$ are revealed one at a time in increasing order along with partial hyperedges $e\cap [j]$ $(e\in E, j\in [n])$.\footnote{According to this definition, hyperedges can have size $1$, although such hyperedges can and will be omitted as they pose no difference with respect to the objective function.} We shall name this problem \textsc{Online Makespan Hypergraph Coloring}$(m)$ and address it in Section \ref{sec:hypergraphs}. Our construction in Section \ref{sec:crazysec} is even restricted to \emph{hyperforests}. A \emph{hyperforest} is a hypergraph without a hypercycle, where a hypercycle is defined as a sequence $v_1,e_1,v_2\ldots, e_\ell,v_{\ell+1}=v_1$ with pairwise distinct hyperedges $e_t$ such that $v_t, v_{t+1}\in e_t$ for all $t\in [\ell]$.

Finally, it is worth mentioning some big numbers that we encounter in Section \ref{sec:crazysec}. In particular, we use Knuth's arrow notation to denote $F(x)\coloneqq x\uparrow\uparrow 5\coloneqq x^{x^{x^{x^{x}}}}$ and the superlogarithm $\mathrm{slog}_5(x)\coloneqq F^{-1}(x)$.

\subsection{Related Work}
The problem that we introduce in this paper lies in the intersection of several well-known problems, which we would like to mention briefly.

\subparagraph{Optimization under Scenarios.} In recent years, scenario-based models have received significant attention. An instance of a scenario model over a discrete optimization problem is often given as an instance of the underlying problem; and in addition, subsets of the underlying instance (called \emph{scenarios}) are specified. The objective function involves the objective values that are attained restricted to each scenario; e.g., the maximum of these values over scenarios or their average.

Although a significant portion of the literature is rather recent, some well-known problems have been scenario-based problems in disguise. For instance, the famous exact matching problem (\cite{exactmatching}), which asks whether a given graph whose edges are colored blue or red has a perfect matching with exactly $k$ blue edges for a given $k$, is equivalent to the matching problem under two scenarios, whereas a stochastic scenario-based model for matchings has recently been studied in~\cite{stocscenmatching}. Other examples include bin packing (\cite{binpacking}\cite{binpacking2}), metric spanning tree (\cite{mst}) and the traveling salesperson problem (\cite{Ee2018priori}). 

Among the scenario-based problems with the widest array of studies is machine scheduling with its many variations, see \cite{shabtay2022} for an extensive overview. The closest to our problem is \cite{feuerstein2016minimizing}, in which the worst-case makespan is being minimized for parallel machine scheduling on two machines and a variable number of scenarios. Although an instance under a variable number of scenarios is not approximable with a ratio better than the trivial bound of $2$, a constant number of scenarios allows a PTAS.

In some variants of scenario scheduling, e.g., in \cite{bosman2023total}, a complexity gap is observed between two and three scenarios. Interestingly, we obtain a distinction between two and three scenarios as well, at least when we restrict to the special case of two machines.

\subparagraph{Online Scheduling.}
As it is impossible to capture the vast amount of literature on all online scheduling problems, let us restrict to online parallel machine scheduling with the objective of minimizing the makespan. This is precisely the special case $K=1$ of our problem. 

The problem was first addressed in Graham's seminal work \cite{graham} from which we also draw inspiration in multiple ways. The competitive ratio of $2-\nicefrac{1}{m}$ for $m$ machines was not beaten for three decades, until it was improved slightly in \cite{galamboswoeginger}. Graham's List Scheduling Algorithm is evidently best possible for $m=2$, and it was observed to be best possible for $m=3$ as well \cite{grahamlb}.
In the following years, a stream of research narrowed the gap between upper and lower bounds for competitive ratios depending on $m$ (\cite{sched1}\cite{sched2}\cite{sched3}), although for $m\geq 4$, still no tight bound is known. The state-of-the-art results are Fleischer and Wahl's $1.9201$-competitive algorithm  (\cite{fleischerwahl}) and an asymptotic lower bound of $1.88$ due to Rudin III (\cite{rudinphd}). For $m=4$, the best known lower bound is $\sqrt{3}\approx 1.732$, due to Rudin III and Chandrasekaran (\cite{m4sqrt3}).

\subparagraph{Discrepancy Minimization, Hypergraph Coloring.}\label{sec:disc}

A problem similar to \textsc{Online Makespan Hypergraph Coloring}, often called Online Hypergraph Coloring, has been studied extensively over the years. In this problem, the objective is to use the smallest number of colors while avoiding a \emph{monochromatic hyperedge}, i.e., a hyperedge whose nodes all attain the same color.  For hypertrees with $n$ nodes, an $O(\log(n))$-competitive algorithm is known to be best possible (\cite{hypertreecoloring}), while for general graphs with $n$ nodes, no $o(n)$-competitive algorithm can exist (\cite{NAGYGYORGY200823}). However, this standard model differs from ours, as a hyperedge is revealed only once all of its nodes have been revealed. This nuance deems it much easier to construct families of instances to achieve lower bounds. Moreover, due to differences in their respective objective functions, this problem does not seem to have a direct relation to \textsc{Online Makespan Hypergraph Coloring}. However, our approach conveniently enables implications for a variant of Online Hypergraph Coloring in which partial hyperedges are known (see Section \ref{contribution}).

The reader may also notice similarities to the well-known discrepancy minimization problem (see e.g.~\cite{sixstddeviations}\cite{tightdisc}): Indeed, in the usual setting of discrepancy minimization, we have $m=2$ and the objective is to minimize the difference $\max_{e\in E}||\{c^{-1}(1)\}\cap e|-|\{c^{-1}(2)\}\cap e||.$ If all hyperedges $e\in E$ have the same size, this objective function is even equivalent to ours. In this regard, the problem that we study generalizes the $\ell_\infty$ norm online discrepancy minimization \cite{onlinedisc} in uniform set systems as well.

\newpage
\section{Extending Graham's List Scheduling}\label{sec:pigeon}

Our first observation is a simple algorithm that generalizes Graham's List Scheduling Algorithm for \textsc{OMSS}$(m,K)$ and performs well when there are too many machines to force the schedule to create significant disbalance. Our algorithm beats the trivial bound of $m$ if $m>K+1$. 

We first observe that if the number $m$ of machines is large enough, we can apply the pigeonhole principle to find a machine that is not among the most loaded machines of any scenario.  For $s\in [m-1]$ and $k\in [K]$, we call a machine $i$ \emph{$s$-favorable with respect to the $k$-th scenario} if there are at least $s$ machines $i'\neq i$ that have at least the load of machine $i$ in the $k$-th scenario, i.e.,  $p(J_{i'}\cap S_k)\geq p(J_i\cap S_k)$.
 
\begin{obs}
  There exists a machine $i$ that is $\left(\left\lceil\frac{m}{K}\right\rceil-1\right)$-favorable with respect to every scenario $k\in [K]$.
\end{obs}
\begin{proof}
There are at most $s$ machines per scenario that are not $s$-favorable. Let $s=\left\lceil\frac{m}{K}\right\rceil-1$.
Since $m>s\cdot K$, there must be a machine $i$ which is not non-$s$-favorable with respect to any scenario.
\end{proof}

In light of this observation, we propose the following algorithm: Upon the assignment of each job $j$, we find a machine $i$ that is $\left(\left\lceil\nicefrac{m}{K}\right\rceil-1\right)$-favorable with respect to all scenarios $k$ with $j\in S_k$ and assign the job $j$ to the machine $i$.

\begin{theorem}\label{thm:grahamgeneral}
  The above algorithm is $\left(\frac{m-1}{\left\lceil\frac{m}{K}\right\rceil}+1\right)$-competitive for $m>K$. 
\end{theorem}
\begin{proof}
Let $i\in [m]$ be a machine and $k\in [K]$ a scenario that attains the makespan of the schedule, and let $j\in [n]$ be the smallest index such that $p(J_i\cap S_k\cap [j])=\mathrm{makespan}(n)$. We define $L\coloneqq p(J_i\cap S_k\cap [j-1])=\mathrm{makespan}(n)-p_j$. By the choice of $i$ upon the insertion of $j$, there are at least $\left\lceil\frac{m}{K}\right\rceil$ machines $i'$ with $p(J_{i'}\cap S_k\cap [j-1])\geq L$, including $i$, so that the offline optimum is bounded from below by
\begin{align*}
    \offl&\geq \max\left\{p_j, \frac{p(S_k)}{m}\right\}\geq \max\left\{p_j, \frac{p(S_k\cap [j-1])+p_j}{m}\right\}&\geq\max\left\{p_j, \frac{\left\lceil\frac{m}{K}\right\rceil\cdot L+p_j}{m}\right\}.
\end{align*}
Hence, for $p_j\geq \frac{\left\lceil\frac{m}{K}\right\rceil}{m-1}L$, the competitive ratio $\rho$ is bounded by
\[\rho\leq \frac{L+p_j}{p_j}\leq 1+\frac{m-1}{\left\lceil\frac{m}{K}\right\rceil},\]
and otherwise it is also bounded by
\[\rho\leq\frac{L+p_j}{\frac{\left\lceil\frac{m}{K}\right\rceil\cdot L+p_j}{m}}=m\cdot\frac{L+p_j}{\left\lceil\frac{m}{K}\right\rceil\cdot L+p_j}\leq m\cdot\frac{1+\frac{\left\lceil\frac{m}{K}\right\rceil}{m-1}}{\left\lceil\frac{m}{K}\right\rceil+ \frac{\left\lceil\frac{m}{K}\right\rceil}{m-1}}= 1+\frac{m-1}{\left\lceil\frac{m}{K}\right\rceil}.\]
Here, the second inequality holds since $\left\lceil\frac m K\right\rceil> 1$.
\end{proof}
\begin{corollary}\label{cor:k1competitive}
    The above algorithm is $(K+1)$-competitive. 
\end{corollary}

This algorithm is also applicable in the special case where every job appears in at most $k$ scenarios for a fixed $k\in \mathbb N$, yielding a ratio of $\frac{m-1}{\left\lceil\frac{m}{k}\right\rceil}+1$. Furthermore, plugging in $K=1$ yields the well-known competitive ratio of $2-\frac{1}{m}$ of Graham's List Scheduling Algorithm.

\section{Near-Tight Bounds for $m=2$ Machines}
When the partitioning is among two machines, a clear line can be drawn: Any algorithm is trivially $2$-competitive, and if we have at least $3$ scenarios, this is best possible.

\begin{obs}\label{thm:lb2K}
    For $m=2$ and $K\geq3$, there exists no $r$-competitive algorithm for any $r<2$, even for the special case of unit processing times.
\end{obs}
\begin{proof}
For all jobs $j$ below, we assume $p_j=1$.    
    The first instance $\mathcal{I}_1$ is given as follows: $n=m+1$, $S_1=\{1,3,\ldots,m+1\}$, $S_2=\{2,3,\ldots, m+1\}$. 
    
    We consider an arbitrary algorithm that places the first jobs $1$ and $2$ to two distinct machines, without loss of generality, to the first and second, respectively. In this case, the makespan at the end must be $2$, as two jobs $(j,j')\neq (1,2)$ must be assigned to the same machine. However, the following assignment, which places the first two jobs together would have yielded a completion time of $1$:    
    $J_1=\{1,2\}, J_i=\{i+1\} \text{ for } 2\leq i\leq m.$
Hence, no algorithm that places two unit weight jobs in $S_1\setminus S_2$ resp.~$S_2\setminus S_1$ to two separate machines can attain a competitive ratio better than $2$.
    
    In light of this, we consider another instance $\mathcal I_2$, on which assigning the first two jobs to the same machine leads to a failure. It is given by $n=m+2$, $S_1=\{1,3\}$, $S_2=\{2,4,5,6,\ldots, m+2\}$, $S_3=\{3,4,5,\ldots, m+2\}$. The first two jobs are indeed identical to those of $\mathcal I_1$.
    
By our considerations about the instance $\mathcal I_1$, we may assume that the first two jobs are assigned to the same machine, the first one without loss of generality.    
However, any assignment that places the first two jobs to the same machine admits makespan $2$: Otherwise, due to the third scenario, the jobs $\{3,\ldots, m+2\}$ must occupy one machine each. The job $j$ that is placed to the first machine has a common scenario with either job $1$ (if $j=3$), or with job $2$ (otherwise).
\end{proof}

On the other hand, this statement cannot be extended to $K=2$. To show this, we present a minimalistic yet insightful algorithm (Algorithm \ref{alg:53}) for makespan scheduling under $K=2$ scenarios and $m=2$ machines. To this end, we combine ideas about what the algorithm can feature with observations about how a worst-case instance is bound to look under the ideas we have fixed. To this end, we first need to introduce a few notions and machinery.

\subsection{Preparation for the Algorithm}
\subparagraph{Fixing a rule.} Our general idea is to fix a simple rule for assigning \emph{double-scenario jobs}, i.e., jobs in $S_1\cap S_2$, and partition the \emph{single-scenario jobs} $j\in S_1\triangle S_2$ so that the machines are reasonably balanced even after upcoming hypothetical double-scenario jobs.

\begin{rul}\label{asp:rulefix}
If $j\in S_1\cap S_2$, then assign the job $j$ to a machine $i$ such that the makespan of the schedule was previously attained by the other machine $3-i$. In case of a tie, select the machine that received the job $j-1$ ($j=1$ is assigned to the first machine if applicable).    
\end{rul}

\subparagraph{A more approachable notion of competitiveness.} Throughout our analysis, we only use two types of lower bounds for the offline optimum: Processing times $p_j$ of selected jobs $j$ whose completion times equal the makespan, and the average load of a selected scenario $S_k$. Hence, while searching for the worst-case sequence of the upcoming double-scenario jobs, we may replace the actual definition of competitiveness with that of a proxy one:

\begin{defn}
    We define the \emph{proxy competitive ratio} of a schedule as
    \[\rho=\frac{\max_{i,k\in\{1,2\}}p(S_k\cap J_i)}{\max\left\{\max_{k\in \{1,2\}}\left\{\frac{p(S_k)}{2}\right\}, \max_{j\in [n]\cap \mathrm{argmax\{C_j\}}}\{p_j\}\right\}},\] 
    where $C_j$ denotes the completion time of the job $j$.
\end{defn}

The proxy competitive ratio is an upper bound on the competitive ratio. The lower bounds that we use are already well-known; for instance, Graham's List Scheduling Algorithm is in fact analyzed by bounding the proxy competitive ratio from above. We take this approach, however, one step further and restrict ourselves to analyzing only the instances that attain the largest possible proxy competitive ratio. This would have been significantly more difficult for the competitive ratio in the usual sense, as a witness to having good or bad offline optima highly depends on how the processing times of jobs fit together numerically. 

\subparagraph{Transforming to a worst-case instance.} In search of instances with the worst possible proxy competitive ratio, we transform arbitrary instances to those that admit a certain structure, and are no better in terms of competitiveness, through a few simple operations.

\begin{defn}
    Let $\mathcal I=((n, (p_j)_{j\in [n]}, \{S_k\}_{k\in [K]}))$ be an instance of \textsc{OMSS$(2,2)$}.  The \emph{deletion} of a job $j$ is the sole modification that $p_j$ is set to $0$, i.e., it gives rise to an instance $\mathcal I'=(n, (p'_j)_{j\in [n]}, \{S_k\}_{k\in [K]})$ with 
        \[p'_{j'}=
        \begin{cases}
        0, & j=j'\\
        p_j, &j\neq j'.
        \end{cases}\]
\end{defn}

\begin{lemma}[Deletion lemma, cf.~Figure \ref{fig:deletionlemma}]\label{lemma:caseadeletion}
Let $\mathcal I$ be an instance on which an algorithm $\mathcal A$ admits a competitive ratio $\rho$.
Let $j\in [n]$ be the largest index such that $ j\in S_1\triangle S_2$ holds. Assume furthermore that the assignment of the $j$-th job to the $i=\tau(j)$-th machine was executed so that the scenario more represented in the machine did not change after the $j$-th job, i.e.,
\[\mathrm{argmax}_{k\in \{1,2\}}p(J_i\cap S_k\cap [j-1])=\mathrm{argmax}_{k\in \{1,2\}}p(J_i\cap S_k\cap [j]).\]
Then the instance $\mathcal I'$ obtained by the deletion of $j$ admits a (proxy) competitive ratio $\rho'\geq \rho.$
\end{lemma}
\begin{proof}
    Let $\tau, \tau'\colon [n]\to\{1,2\}$ denote the assignment of the original instance and and the modified instance, respectively. Both instances admit the identical assignment up to (and including) job $j-1$. We observe that
    \[\mathrm{argmax}_{k\in \{1,2\}}p(J_i\cap S_k\cap [j-1])=\mathrm{argmax}_{k\in \{1,2\}}p(J_i\cap S_k)\]
    for $i\in \{1,2\}$ for $\mathcal I$ and analogously for $\mathcal I'$, since $\{j+1,\ldots, n\}\subseteq S_1\cap S_2$ by the choice of $j$ and therefore the difference $p(J_i\cap S_1)-p(J_1\cap S_2)$ stays constant after the assignment of the $j$-th job. For both instances, the jobs $j+1,\ldots, n$ are distributed to a machine $i$ such that the other machine $3-i$ admits makespan by the same reasoning. Since the assignment of the job $j$ does not affect $\max_{k\in \{1,2\}}p(J_i\cap S_k\cap [j'])$ for any $j'\in [n]$, we have $\tau(j')=\tau'(j')$ for $j'>j$. 
    
    The $j$-th job in instance $\mathcal{I}$ cannot be in a combination of machine and scenario that admits the makespan of the schedule. Thus the makespan $\onl_{\mathcal{I}}$ and $\onl_{\mathcal I'}$ of both instances coincide. 
        On the other hand, the offline optimal value or the denominator of the proxy competitive ratio do not increase by the modification, so that $\rho'\geq\rho$ holds.
\end{proof}
\begin{figure}
    \centering
    \begin{tikzpicture}
             \draw [fill=blue] (0,0) rectangle (2,0.5) node[pos=.5] {};
              \draw [fill=blue] (0,0.5) rectangle (1,1) node[pos=.5] {};
            \draw [fill=blue] (0,1.2) rectangle (1.3,1.7) node[pos=.5] {};
             \draw [myred] (1.3,1.2) rectangle (2.5,1.7) node[pos=.5] {{$\mathbf{j}$}};
             \draw [fill=blue] (0,1.7) rectangle (3.5,2.2) node[pos=.5] {};
                \draw [fill=blue] (1.5,1.7) rectangle (3.5,2.2) node[pos=.5] {};
                \draw [fill=blue!30] (2,0) rectangle (4,0.5) node[pos=.5] {$n-1$};
                 \draw [fill=blue!30] (1,0.5) rectangle (3,1) node[pos=.5] {$n-1$};
                    \draw [fill=blue!30] (4,0) rectangle (5.5,0.5) node[pos=.5] {$n$};
                         \draw [fill=blue!30] (3,1) rectangle (4.5,0.5) node[pos=.5] {$n$};
           
              \node at(-1, 0.2){$J_2\cap S_2$};
              \node at(-1, 0.7){$J_2\cap S_1$};
              \node at(-1, 1.4){$J_1\cap S_2$};
              \node at(-1, 1.9){$J_1\cap S_1$};
     
\end{tikzpicture}

    \caption{Deletion of the largest $j\in S_1\triangle S_2$ (checkered) does not increase the approximation ratio, as it does not affect the execution of Rule \ref{asp:rulefix}.}
    \label{fig:deletionlemma}
\end{figure}
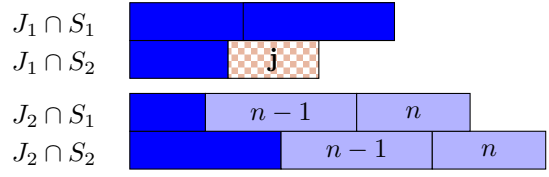
We may assume for the completion time of the largest job index $j\in S_1\triangle S_2$ that $C_j>C_{j-1}$, for otherwise, we may apply Lemma \ref{lemma:caseadeletion} successively to eliminate jobs $j\in S_1\triangle S_2$ that cause a violation of this inequality and analyze the resulting instances instead. 

Next, we observe that under certain circumstances, we may ``cut'' the double-jobs $j$ in that we replace them by other jobs $j',j''\in S_1\cap S_2$ such that $p_{j'}+p_{j''}=p_j$, while ensuring that the behavior of our algorithm remains the same way:

\begin{lemma}[Cutting lemma, cf.~Figure \ref{fig:cuttinglemma}] \label{lemma:cutting}
Let $t\in  S_1\cap S_2$ be a job index for an instance $\mathcal I=(n, (p_j)_{j\in [n]}, \{S'_k\}_{k\in \{1,2\}})$ of \textsc{OMSS(2,2)} such that $\mathrm{makespan}(t)>\mathrm{makespan}(t-1)$. Consider a schedule $\tau\colon [n]\to \{1,2\}$ of $n$ jobs, on which the algorithm admits a (proxy) competitive ratio of $\rho$. Then there is another instance $\mathcal{I}'=(n+1, (p'_j)_{j\in [n]}, \{S'_k\}_{k\in \{1,2\}})$ with $n+1$ jobs with processing times $p'_j, j\in [n+1]$ under scenarios $S_1', S_2'\subseteq [n+1]$ with the following properties:
\begin{enumerate}[(i)]
    \item $p'(J'_i\cap S'_k\cap [j])=p(J_i\cap S_k\cap [j])$ for $j\in [t-1]$, $i\in \{1,2\}$ and $k\in \{1,2\}$.
    \item The algorithm assigns the first $t$ jobs of $\mathcal I'$ so that \[\max_{k\in \{1,2\}}p(J'_1\cap S'_k\cap [t])=\max_{k\in \{1,2\}}p(J'_2\cap S'_k\cap [t]).\] 
    \item The algorithm admits a (proxy) competitive ratio $\rho'\geq \rho$ on the instance $\mathcal I'$.
\end{enumerate}
\end{lemma}
\begin{proof}
    We define the new instance $\mathcal I'$ as follows: we set   
    \[S'_k=(S_k\cap [t])\;\dot \cup \;[t+1]\;\dot \cup\; \{j\in \{t+2,\ldots,n+1\}\colon j-1\in S_k\}\] for $t\in [n]$ and $k\in \{1,2\}$. Further, we define $p'_j=p_j$ for $1\leq j<t$ and $p'_j=p_{j-1}$ for $t+2\leq j\leq n+1$. We determine $p'_t$ and $p'_{t+1}$ later. We will aim to satisfy $p'_{t}+p'_{t+1}=p_{t}$.
    
    We start by noticing that the instances $\mathcal I$ and $\mathcal I'$ coincide when restricted to the first $t-1$ jobs. Therefore, the algorithm will yield identical assignments until the assignment of the $t$-th job and the first property is satisfied.

    Now we consider the job indexed $t$. In both instances, we have $t\in S_1\cap S_2$ resp.~$t\in S_1'\cap S_2'$. For such a double-scenario job, the assignment only depends on the prior subschedule  $\tau|_{[t-1]}=\tau'|_{[t-1]}$ and not on the processing time $p_t$ resp.~$p_t'$. Therefore, we also have $\tau(t)=\tau(t')$ for any $p'_t\in (0, p_t)$. To obtain the second property, we pick 
    \begin{equation}\label{eq:minhere}
           p'_t=\mathrm{makespan}(t-1)
-\min_{i'\in \{1,2\}}\max_{k'\in \{1,2\}}p(J'_{i'}\cap S'_k\cap [t-1]),  
    \end{equation}

    and $p'_{t+1}=p_t-p'_t$. We guarantee that $p'_{t+1}>0$ because of $\mathrm{makespan}(t)>\mathrm{makespan}(t-1)$. Since the $t$-th job is assigned to the $i'$-th machine attaining the minimum in \eqref{eq:minhere}, we then have
    \begin{equation}\label{eq:formermakespan}
         \max_{k\in \{1,2\}}p'(J'_{i'}\cap S'_k\cap [t])=\max_{k\in \{1,2\}}p'(J'_{i'}\cap S'_k\cap [t-1])+p'_t=\mathrm{makespan}(t-1).
    \end{equation}
   
   Since the $t$-th job was not assigned to the $(3-i')$-th machine, we also have
    \begin{equation}\label{eq:newmakespan}
        \max_{k\in \{1,2\}}\{p'(J'_{3-i'}\cap S'_k\cap [t])\}= \max_{k\in \{1,2\}}\{p'(J'_{3-i'}\cap S'_k\cap [t-1])\}=\mathrm{makespan}(t-1).
    \end{equation}
    Comparing lines \eqref{eq:formermakespan} and \eqref{eq:newmakespan} reveals that the second property is satisfied.

    For the third property, we observe that both machines admit makespan prior to the insertion of the $(t+1)$-st job, so the algorithm will assign this job to the machine where it has assigned the previous job $t$ to. This yields
    \begin{align*}
        p'(J'_{i'}\cap S'_k\cap [t+1])&=p'(J'_{i'}\cap S'_k\cap [t-1])+p'_t+p'_{t+1}\\ 
        &=p(J_{i'}\cap S_k\cap [t-1])+p_t=p(J_{i'}\cap S_k\cap [t])
    \end{align*}
    for $k\in \{1,2\}$. 
    Since the $j$-th job in $\mathcal{I}$ is identical to the $(j+1)$-st job in $\mathcal I'$, the algorithm will output the same value in both cases. It is also straightforward to observe that the denominator of the proxy competitive ratio cannot have increased after the modification.
    
   In total, this implies
    $\rho'\geq \rho,$
    proving that the third property is satisfied.
 
\end{proof}

\begin{figure}
    \centering
    \begin{tikzpicture}
             \draw [fill=blue] (0,0) rectangle (2,0.5) node[pos=.5] {};
              \draw [fill=blue] (0,0.5) rectangle (1,1) node[pos=.5] {};
            \draw [fill=blue] (0,1.2) rectangle (1.3,1.7) node[pos=.5] {};
             \draw [fill=blue] (0,1.7) rectangle (4,2.2) node[pos=.5] {};
                \draw [fill=blue!30] (2,0) rectangle (3.5,0.5) node[pos=.5] {$n-1$};
                 \draw [fill=blue!30] (1,0.5) rectangle (2.5,1) node[pos=.5] {$n-1$};
                    \draw [fill=blue!30] (3.5,0) rectangle (5.5,0.5) node[pos=.5] {$n$};
                         \draw [fill=blue!30] (2.5,1) rectangle (4.5,0.5) node[pos=.5] {$n$};
           
              \node at(-1, 0.2){$J_2\cap S_2$};
              \node at(-1, 0.7){$J_2\cap S_1$};
              \node at(-1, 1.4){$J_1\cap S_2$};
              \node at(-1, 1.9){$J_1\cap S_1$};

                      \draw [fill=blue] (6.5,0) rectangle (8.5,0.5) node[pos=.5] {};
              \draw [fill=blue] (6.5,0.5) rectangle (7.5,1) node[pos=.5] {};
            \draw [fill=blue] (6.5,1.2) rectangle (7.8,1.7) node[pos=.5] {};
             \draw [fill=blue] (6.5,1.7) rectangle (10.5,2.2) node[pos=.5] {};
                \draw [fill=blue!30] (8.5,0) rectangle (10,0.5) node[pos=.5] {$n-1$};
                 \draw [fill=blue!30] (7.5,0.5) rectangle (9,1) node[pos=.5] {$n-1$};
                    \draw [fill=blue!30] (10,0) rectangle (10.5,0.5) node[pos=.5] {$n$};
                         \draw [fill=blue!30] (9,1) rectangle (9.5,0.5) node[pos=.5] {$n$};  
                               \draw [fill=blue!30] (10.5,0) rectangle (12,0.5) node[pos=.5] {$n+1$};
                         \draw [fill=blue!30] (9.5,1) rectangle (11,0.5) node[pos=.5] {$n+1$};  
                         
              \draw[dashed] (10.5,-0.5)--(10.5,3);
               \node at(11.5, 2.7){$C_j=C_{n}$};
\end{tikzpicture}

    \caption{Execution of the cutting lemma on job $n$, before (left) and after (right).}
    \label{fig:cuttinglemma}
\end{figure}
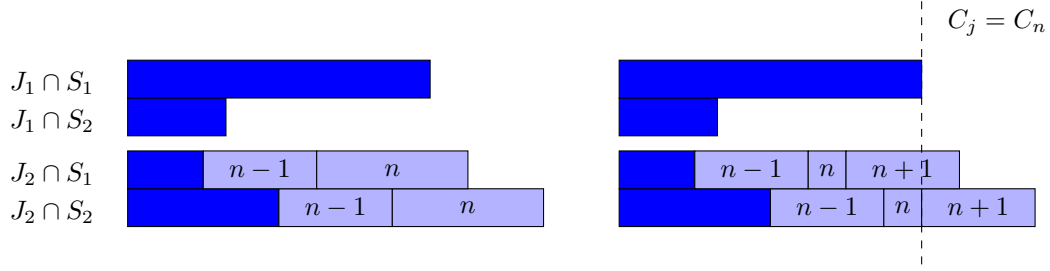

We observe that in order to analyze the proxy competitive ratio of instances whose last job $n$ satisfies $n\in S_1\cap S_2$, it suffices to consider those instances,
in which additionally the first $n-1$ jobs satisfy 
\begin{equation}\label{eq:cuttinglemmaeq}
 p(S_k\cap J_k\cap[n-1])=p(S_{3-k}\cap J_{3-k}\cap [n-1]).    
\end{equation}
Indeed, if we have an arbitrary instance with last job $n\in S_1\cap S_2$, we can apply Lemma \ref{lemma:cutting} with $t=n$. 

On the other hand, if \eqref{eq:cuttinglemmaeq} is satisfied, the adversarial last job $n$ is double-scenario and has processing time \[p_n=p(S_k\cap J_k\cap[n-1])+\max_{k\in \{1,2\}}p(S_k\cap J_{3-k}\cap [n-1])=\max_{k\in \{1,2\}}p(S_k\cap[n-1])\], so as to maximize the proxy competitive ratio
\begin{align}
\rho&=\frac{p(S_k\cap J_k\cap[n-1])+p_n}{\max\left\{\max_{k\in \{1,2\}}\left\{\frac{p(S_k\cap[n-1])+p_n}{2}\right\}, p_n\right\}}\\&=\frac{2p(S_k\cap J_k\cap [n-1])+\max_{k'\in \{1,2\}}p(S_{k'}\cap J_{3-k'}\cap [n-1])}{p(S_k\cap J_k\cap [n-1])+\max_{k'\in \{1,2\}}p(S_{k'}\cap J_{3-k'}\cap [n-1])}\\&\label{eq:adversarialp}=1+\frac{p(S_k\cap J_k\cap [n-1])}{\max_{k'\in \{1,2\}}p(S_{k'}\cap J_{3-k'}\cap [n-1])}.
\end{align}
Now let $j\in [n]$ be the highest-indexed job that is in $S_1\triangle S_2$. If $j=n$, we assume to be done, so let $j<n$. By Lemma \ref{lemma:caseadeletion}, there is a job $j'>j$ whose completion time is larger than that of $j$ in some scenario. We can pick the smallest-indexed $j'$ whose completion time in some scenario $S_k$ surpasses the completion time of $j$. For this index $j'$, we apply Lemma \ref{lemma:cutting} so that after the assignment of the earlier of the two jobs that replace $j'$, the completion time of the job $j$ still equals the overall makespan of the schedule. This implies:

\begin{obs}\label{obs:twocuts}
    Among the instances $\mathcal I=(n, (p_j)_{j\in [n]}, \{S_1,S_2\})$ with $n\in S_1\cap S_2$, there is a worst-case instance with the following properties: 
    \begin{enumerate}[(i)]
        \item     If $j\in [n]$ is the largest index with $j\in S_1\triangle S_2$, there is another job indexed $j'>j$ such that $C_j=C_{j'}$.
        \item Equation \eqref{eq:cuttinglemmaeq} holds, i.e.,  $p(S_k\cap J_k\cap[n-1])=p(S_{3-k}\cap J_{3-k}\cap [n-1])$.   
    \end{enumerate}
\end{obs}

We transform this instance further by deleting the jobs $j'+1,\ldots, n-1$ and changing $p_n$ to $p(S_k\cap J_k\cap [j'])+\max_{k'\in \{1,2\}}p(S_{k'}\cap J_{3-k'}\cap [j'])$. Since $j'$ is the penultimate job after the modification, the equation \eqref{eq:cuttinglemmaeq} holds for the modified instance as well.

If $\rho'$ denotes the new proxy competitive ratio after the modifications, then we have 

\begin{align}
\rho&= 1+\frac{p(S_k\cap J_k\cap [n-1])}{\max_{k'\in \{1,2\}}p(S_{k'}\cap J_{3-k'}\cap [n-1])}\\&= 1+\frac{p(S_k\cap J_k\cap [j'])+p(J_k\cap S_k\cap \{j'+1,\ldots, n-1\})}{\max_{k'\in \{1,2\}}(p(S_{k'}\cap J_{3-k'}\cap [j'])+p(S_{k'}\cap J_{3-k'}\cap \{j'+1\ldots, n-1\}))}
\\&=1+\frac{p(S_k\cap J_k\cap [j'])+p(J_k\cap \{j'+1,\ldots, n-1\})}{\max_{k'\in \{1,2\}}(p(S_{k'}\cap J_{3-k'}\cap [j'])+p(J_{3-k'}\cap \{j'+1\ldots, n-1\}))}
\\&\stackrel{(*)}\leq 1+\frac{p(S_k\cap J_k\cap [j']))}{\max_{k'\in \{1,2\}}p(S_{k'}\cap J_{3-k'}\cap [j'])}=\rho'.
\end{align}
Here, the inequality $(*)$ follows by Observation \ref{obs:twocuts}. Since $\rho$ is the proxy competitive ratio of a worst-case instance among those with $n\in S_1\cap S_2$, we conclude that the modified instance is a worst-case instance as well. We can therefore further strengthen the properties sketched in Observation \ref{obs:twocuts} by assuming that the jobs before the penultimate job $j'=n-1$ and after the last job $j\in S_1\triangle S_2$ had smaller completion time than the job $j$.
In summary, we can observe the following statement.

\begin{obs}[cf.~Figure \ref{fig:schedule}]\label{obs:bottleneck}
        Let $\mathcal A$ be an online algorithm for \textsc{OMSS$(2,2)$} obeying Rule \ref{asp:rulefix}. Assume that $\mathcal{A}$ has a proxy competitive ratio of at most $\frac{5}{3}$ on instances where
        \begin{enumerate}[(i)]
            \item $n\in S_1\cap S_2$, and
            \item $\mathcal A$ outputs a schedule with the following property:   If $j\in [n]$ is the largest index with $j\in S_1\triangle S_2$, then $C_j=C_{n-1}$.
        \end{enumerate}
        Then, $\mathcal A$ is $\nicefrac{5}{3}$-competitive.
\end{obs}

\begin{figure}
    \centering
    \begin{tikzpicture}
             \draw [fill=blue] (0,0) rectangle (2,0.5) node[pos=.5] {};
              \draw [fill=blue] (0,0.5) rectangle (1,1) node[pos=.5] {};
            \draw [fill=blue] (0,1.2) rectangle (1,1.7) node[pos=.5] {};
             \draw [fill=blue] (0,1.7) rectangle (3.5,2.2) node[pos=.5] {};
                \draw [fill=blue] (1.5,1.7) rectangle (3.5,2.2) node[pos=.5] {\color{white}{$\mathbf{j}$}};
                \draw [fill=blue!30] (2,0) rectangle (3.5,0.5) node[pos=.5] {$n-1$};
                 \draw [fill=blue!30] (1,0.5) rectangle (2.5,1) node[pos=.5] {$n-1$};
                    \draw [myred] (3.5,0) rectangle (9.5,0.5) node[pos=.5] {$n$};
                         \draw [myred] (2.5,1) rectangle (8.5,0.5) node[pos=.5] {$n$};
           
              \node at(-1, 0.2){$J_2\cap S_2$};
              \node at(-1, 0.7){$J_2\cap S_1$};
              \node at(-1, 1.4){$J_1\cap S_2$};
              \node at(-1, 1.9){$J_1\cap S_1$};
              \draw[dashed] (3.5,-0.5)--(3.5,3);
               \node at(4.5, 2.7){$C_j=C_{n-1}$};
     
\end{tikzpicture}

    \caption{Illustrative description of Observation \ref{obs:bottleneck}. Given the dark blue schedule, a worst outcome turns out to be the light blue job $n-1$ with a processing time $p_{n-1}$ such that $C_{n-1}=C_j$, followed by the light red job $n$ with $p_n=p(S_1\cap [n-1])$. The anticipation measures the minimum discrepancy in either machine prior to the revelation of the final job $n$.}
    \label{fig:schedule}
\end{figure}
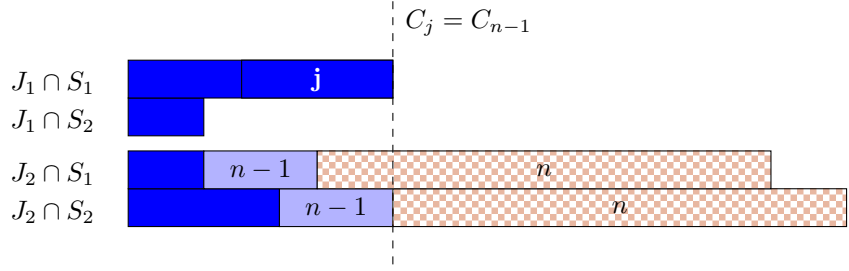

With Lemma \ref{lemma:caseadeletion} in mind, we can conclude: among instances with $n\in S_1\cap S_2$, at least one worst-case instance behaves as follows: The last single-scenario job $j\in S_1\triangle S_2$ admits the makespan at the time of its insertion. Then, the double-scenario jobs $j+1,\ldots, n-1$ are all assigned to the opposite machine. After the insertion of the $(n-1)$-st job, both machines have the same makespan; i.e., the schedule satisfies \eqref{eq:cuttinglemmaeq}. The last job $n$ has the adversarial processing time of \[p_n=\max_{k\in \{1,2\}}p(S_k\cap [n-1])=p(S_k\cap J_k\cap[n-1])+\max_{k\in \{1,2\}}p(S_k\cap J_{3-k}\cap [n-1]).\] Figure \ref{fig:schedule} illustrates a ``worst'' possible outcome, given a partial schedule.

\subparagraph{Anticipation.} In the sequel, we assume that the instance which we analyze satisfies the properties in Observation \ref{obs:bottleneck}. For such a special instance, a key parameter turns out to capture discrepancies within each machine.

\begin{defn}
    Consider a schedule where, without loss of generality, the makespan is attained at $S_1\cap J_1$. The \emph{anticipation} $\alpha$ of this schedule is given by the ratio \[\alpha\coloneqq\begin{cases}
        \frac{p(J_1\cap S_1)}{\max\left\{p(J_1\cap S_2),p(J_2\cap S_1)+p(J_1\cap S_1)-p(J_2\cap S_2)\right\}} &\text{ if }p(J_2\cap S_2)>p(J_2\cap S_1)\\
        0 &\text{ else.}
    \end{cases}\] By convention, we assume $\frac 0 0 =1$.
\end{defn}

In Figure \ref{fig:schedule}, we have $p(J_2\cap S_2\cap [n-2])>p(J_2\cap S_1\cap [n-2])$ and the second term of the denominator corresponds to the load $p(J_2\cap S_1\cap [n-1])$, i.e., the load of $J_2\cap S_1$ after the $(n-1)$-st job is assigned. To further motivate the concept of anticipation, let us bound the proxy competitive ratio in terms of it.
\begin{lemma}\label{motivateanticip}
    Let a schedule of $n-1$ jobs be given that satisfies \eqref{eq:cuttinglemmaeq} and anticipation $\alpha$. After any assignment of the $n$-th job, the resulting schedule has proxy competitive ratio at most $\frac{2\alpha+1}{\alpha+1}$.
\end{lemma}
\begin{proof}
   Let $P\coloneqq p(S_1\cap J_1\cap [n-1])$. Then if $p_n\leq P+\frac{P}{\alpha}$, we have 
   \[ \rho\leq \frac{2P+2p_n}{P+{\max\{p(S_1\cap J_1),p(S_2\cap J_2)\}+p_n}}\leq \frac{2P+2p_n}{P+\frac P\alpha+P\left(1+\frac 1 \alpha\right)}\leq \frac{2\alpha+1}{\alpha+1},\]
   and otherwise,
   \[\rho\leq \frac{P+p_n}{p_n}\leq\frac{P+\left(1+\frac{1}{\alpha}\right)P}{\left(1+\frac{1}{\alpha}\right)P}=\frac{2\alpha+1}{\alpha+1},\]
   showing the claim.
\end{proof}

\subsection{The Algorithm} In our algorithm, one of the goals is to keep $\alpha\leq 2$ in order to brace ourselves for upcoming double-scenario jobs in order to ensure a proxy competitive ratio of at most $\nicefrac{5}{3}$. The invariant which we maintain to this end is the following.

\begin{inv}\label{asp:induction}
  The schedule satisfies the following assertions:
    \begin{enumerate}[(i)]
        \item Its proxy competitive ratio is bounded by $\rho \leq \nicefrac{5}{3}$.
        \item The schedule has anticipation bounded by $\alpha \leq 2$.
        \item\label{item:ppty3}  Let $k,i$ be given so that $J_i\cap S_k$ admits the makespan in the schedule. If the schedule is \emph{dominated by the $i$-th machine}, i.e., $p(J_i\cap S_{k'})> p(J_{3-i}\cap S_{k''})$ for all $k',k''\in \{1,2\}$, then we have $p(J_i\cap S_k)\leq 2p(J_i\cap S_{3-k})$.
    \end{enumerate}
\end{inv}

The last property might appear somewhat unnatural at the first glance, though it is crucial for bounding the anticipation $\alpha \leq 2$ in the subsequent iterations.

Algorithm \ref{alg:53} shows how a single job $j$ is assigned given a schedule of the first $j-1$ jobs where Invariant \ref{asp:induction} is maintained. In the sequel, we prove that instances satisfying the properties in Observation \ref{obs:bottleneck} are always assigned so that the first $n-1$ jobs maintain Invariant \ref{asp:induction} and the last job results in a proxy competitive ratio of at most $\nicefrac{5}{3}$.

\begin{algorithm}
\caption{The assignment of a single job $j$ given a $\frac{5}{3}$-competitive schedule that satisfies Invariant \ref{asp:induction}.}
\label{alg:53}
\begin{algorithmic}[1]
\State Rename machines and scenarios so that $(J_1\cap S_1)$ is maximal
        \If{$j\in S_1\setminus S_2$}
            \State \textbf{return} 2
        \EndIf
        \If{$j\in S_2\setminus S_1$}
            \If{assigning $j$ to the first machine satisfies Invariant \ref{asp:induction}}
            \State \textbf{return} 1
            \Else
            \State \textbf{return} 2
            \EndIf
        \EndIf
        \If{$j\in S_1\cap S_2$} \State \textbf{return} $ \mathrm{argmin}_{i=1,2}\max_{k=1,2}\{p(J_i\cap S_k)\}$ 
        \EndIf        
\end{algorithmic}  
\end{algorithm}

\subsection{Analysis of Algorithm \ref{alg:53}}
It is evident that if Line 7 is executed, the assumption is maintained and we obtain a $\nicefrac{5}{3}$-competitive schedule. It is also easy to observe that the execution of Line 3 provides this: In the first scenario, the second machine is a least loaded machine, so that $\rho$ or $\alpha$ cannot increase.\footnote{This is straightforward to observe and left to the reader as an exercise.}
Moreover, if the schedule was previously dominated by the first machine, the third property is maintained as we do not increase $p(J_1\cap S_1)$. 

Finally, we deal with the execution of Line 13: the proxy competitive ratio does not exceed $\nicefrac{5}{3}$ after the last single-scenario job, which we prove below.

\begin{lemma}
    For an instance of OMSS(2,2) satisfying the properties in Observation \ref{obs:bottleneck}, let $j\in [n]$ be the largest index with $j\in S_1\triangle S_2$. Then, if Algorithm \ref{alg:53} satisfies Invariant \ref{asp:induction} on the first $j$ jobs, then it also satisfies Invariant \ref{asp:induction} on the first $\ell$ jobs for $\ell\in \{j+1,\ldots, n-1\}$.
\end{lemma}
\begin{proof}\begin{enumerate}[(i)]
    \item $\alpha\leq 2$: We claim that $\alpha$ stays constant after the insertion of the $j$-th job. Indeed, $p(J_2\cap S_2\cap [\ell])>p(J_2\cap S_1\cap [\ell])$ is maintained (as $p(J_2\cap S_2\cap [\ell])-p(J_2\cap S_1\cap [\ell])$ is constant). Similarly, both $p(J_1\cap S_2\cap [\ell])$ and $p(J_2\cap S_1\cap [\ell])+p(J_1\cap S_1\cap[\ell])-p(J_2\cap S_2\cap [\ell])$ are constant in $\ell\in \{j+1,\ldots, n\}$.
    \item $\rho\leq\nicefrac{5}{3}$: By Lemma \ref{motivateanticip}, it suffices to observe that the anticipation prior to the assignment of the jobs $\ell\in \{j+1,\ldots, n\}$ is constant.
    \item Third property: If there was no machine dominance after the insertion of the $j$-th job, machine dominance does not emerge since the jobs $\ell\in \{j+1,\ldots, n-1\}$ are all assigned to the second machine. Otherwise, the condition of the third property is maintained for the same reason.
\end{enumerate}
\end{proof}
This lemma does not cover the $n$-th job, however, the $n$-th job only needs to satisfy $\rho\leq \nicefrac{5}{3}$, which follows by the fact that $\alpha\leq 2$ prior to its assignment as we have seen, and Lemma \ref{motivateanticip}.

All that remains to analyze is that if Line 9 is executed, Invariant \ref{asp:induction} is maintained. 
Notice that machine dominance cannot occur as the prior makespan is still among the top two largest loads and is in the opposite machine of the makespan. Hence the third property does not need to be checked. 

We apply a case distinction. In each case, we either bound the proxy competitive ratio $\rho \leq \nicefrac{5}{3}$ and the anticipation $\alpha\leq 2$ by means of inequalities, or find inequalities that are necessary for a condition to fail and show that conjunction of these conditions has no solutions. As all inequalities are linear in $p(J_i\cap S_k\cap [j-1])$ and $p_j$, the reader might interpret the remainder of the proof as a sequence of cheerful linear inequalities which they are encouraged to trace. Accordingly, we ease the notation, and define:
\begin{align*}
    x_1&\coloneqq p(J_1\cap S_1\cap [j-1])\\
    x_2&\coloneqq p(J_1\cap S_2\cap [j-1])\\
    x_3&\coloneqq p(J_2\cap S_1\cap [j-1])\\
    x_4&\coloneqq p(J_2\cap S_2\cap [j-1])\\
    x_5&\coloneqq p_j
\end{align*}

\noindent We start with the case that $j$ has a large processing time.

 \subparagraph{Case 1: $x_5\geq 2x_1-x_2$.} In this case, the attempt to assign the job $j$ to the first machine failed because \[p(J_1\cap S_2\cap [j])=x_2+x_5\geq x_2+2x_1-x_2=2x_1,\] so that the assignment to the second machine yields a makespan of 
 \[p(J_2\cap S_2\cap [j])=x_4+x_5\geq x_4+2x_1-x_2\geq x_1.\]

\noindent The proxy competitive ratio is at most 
\[\rho\leq \max\left\{\frac{x_4+x_5}{x_5}, \frac{2x_4+2x_5}{x_2+x_4+x_5}\right\}\leq \frac 5 3.\]
Indeed, if $x_1\geq 2x_2$, then we have 
\[\frac{x_4+x_5}{x_5}\leq \frac{x_4+2x_1-x_2}{2x_1-x_2}\leq \frac{3x_1-x_2}{2x_1-x_2}\leq \frac{2.5x_1}{1.5x_1}=\frac{5}{3},\]
which is equivalent to $2x_5\geq 3x_4$. Otherwise we have 
\[\frac{2x_4+2x_5}{x_2+x_4+x_5}\leq \frac{5x_4}{x_2+2.5x_4}\leq\frac
{5x_1}{0.5x_1+2.5x_1}\leq \frac{5}{3}.\]

\noindent Finally, the anticipation is bounded by 
\[\alpha\leq\frac{x_4+x_5}{x_2+x_4+x_5-x_1}\leq \frac{x_5}{x_2+x_5-x_1}\leq \frac{2x_1-x_2}{x_2+2x_1-x_2-x_1}=\frac{2x_1-x_2}{x_1}\leq 2,\]
so that all properties in Invariant \ref{asp:induction} are shown.

In the sequel, we assume that $x_5\leq 2x_1-x_2$ and distinguish some cases based on how $x_1,x_2,x_3,x_4$ relate to each other. 

\subparagraph{Case 2: $x_1\geq x_2\geq x_3$ and $x_2\geq x_4$.} The $j$-th job was assigned to the second machine. This cannot be because of properties (ii) or (iii), as these could only be violated if $x_5\geq 2x_1-x_2$, which we have already excluded. Hence, assigning the $j$-th job to the first machine must have created a proxy competitive ratio of larger than $\frac 5 3$, i.e.,

\begin{equation}
    \frac{2x_2+2x_5}{x_2+x_4+x_5}>\frac 5 3 \Leftrightarrow x_2+x_5>5x_4.
\end{equation}

If the makespan is still attained by $J_1\cap S_1$, then the proxy competitive ratio $\rho$ has not increased and is bounded by $\frac 5 3 $ by induction. Otherwise the new makespan is $x_4+x_5\geq x_1$ and the proxy ratio is bounded by $\rho \leq\frac{x_4+x_5}{\max\{x_5, \frac{x_2+x_4+x_5}{2}\}}\leq \frac 3 2$, similar to the proof of Graham's List Scheduling Algorithm since $j$ was assigned to a least loaded machine.

We only need to check the property (ii), i.e., that $\alpha \leq 2$. However, the anticipation is bounded by $\frac{x_1}{x_2}\leq 2$ since the initial schedule must satisfy the property (iii).

\subparagraph{Case 3: $x_1\geq x_3$ and $x_4\geq x_2$.} We claim that this case never occurs, that is, in this case the job $j$ is always assigned to the first machine. Since we assume $x_5\leq 2x_1-x_2$, the property (iii) cannot be the reason the assignment to the first machine has failed. Similarly, if $x_5+x_2\geq x_1$, the anticipation is bounded by $\frac{x_5+x_2}{x_1}\leq \frac{2x_1-x_2+x_2}{x_1}=2$, and otherwise the anticipation cannot increase, either; so the property (ii) is also maintained while assigning to the first machine. Finally, we note that by $x_4\geq x_2$, the first machine is least loaded in the second scenario, therefore by the proof of Graham's List Scheduling, the proxy ratio either decreases or is bounded by $\frac 3 2 $.
\subparagraph{Case 4: $x_1\geq x_3\geq x_2\geq x_4$, proof by LP solver.} Exactly as above, the assignment cannot have failed because of the properties (ii) and (iii). Therefore, assigning $j$ to the first machine must yield a proxy competitive ratio larger than $\frac 5 3$, which implies both $2x_5\leq 3x_2$ and $x_5+x_2\leq 5x_4$. 

As $x_4\leq x_2$, the second machine is least loaded in the second scenario and therefore $\rho\leq \frac 5 3$ is secured if we place $j$ to the second machine. By the above observation, we can further bound the anticipation by
\[\alpha\leq \min\left\{\frac{x_4+x_5}{x_3+(x_4+x_5-x_1)}, \frac{x_1}{x_2+(x_1-x_4-x_5)}\right\}.\]
If $x_4+x_5\geq x_1$, the first term is valid, which means
\begin{equation}
 \frac{x_4+x_5}{x_3+(x_4+x_5-x_1)}>2 \Leftrightarrow 2(x_1-x_3)>x_4+x_5.
\end{equation}
Gathering all constraints, we obtain system of inequalities 
\begin{align}
    x_1&\geq x_3\\
    x_3&\geq x_2\\
    x_2&\geq x_4\\
    2x_5&\leq 3x_2\\
    \label{eq:replace1}x_4+x_5&\geq x_1\\
    \label{eq:replace2}2x_1-2x_2&\geq x_5+x_4\\
    x_4+x_5&\geq 2x_3\\
    x_2+x_5&\geq 5x_4
\end{align}
The only feasible solutions to this system of inequalities are of the type $(4x_4, 2x_4, 2x_4, x_4, 3x_4)$. For these solutions, it is immediate that $\alpha\leq 2$, from which the claim follows. 

Similarly if $x_4+x_5\geq x_1$, the inequality \eqref{eq:replace1} is reversed and the inequality \eqref{eq:replace2} is replaced by $\frac{x_1}{x_3+(x_1-x_4-x_5)}>2\Leftrightarrow2x_4+2x_5\geq 2x_3+x_1$, so that all schedules to be inspected are the solutions of the following system: 

\begin{align}
    x_1&\geq x_3\\
    x_3&\geq x_2\\
    x_2&\geq x_4\\
    2x_5&\leq 3x_2\\
    x_4+x_5&\leq x_1\\
    2x_4+2x_5&\geq 2x_3+x_1\\
    x_4+x_5&\geq 2x_3\\
    x_2+x_5&\geq 5x_4
\end{align}

This system of inequalities has the same set of solutions as well, so that the anticipation is bounded by $\alpha \leq 2$.

As all cases are covered, we conclude that Algorithm \ref{alg:53} maintains the properties in Invariant \ref{asp:induction} and is therefore $\nicefrac{5}{3}$-competitive.

\subsection{Lower Bounds}
In spite of the restrictions that arise from both the fixed double-scenario rule and how we evaluate lower bounds, we find out that our algorithm in the previous subsection is not far from being best possible.

\begin{theorem}\label{thm:53lb}
    For $m=2$ and $K=2$, there exists no deterministic algorithm with competitive ratio strictly smaller than $\frac{9+\sqrt{17}}{8}\approx 1.640$.
\end{theorem}
\begin{proof}
    To ease the notation once again, let $X\coloneqq 3+\sqrt{17}$. We demonstrate our family of instances by means of a case distinction, that is, our instances mostly agree with each other in earlier jobs and their later jobs are revealed depending on the schedule. 

The first four jobs are given by $1,2\in S_1\setminus S_2$, $3,4\in S_2\setminus S_1$, $p_1=p_2=p_3=p_4$. Any algorithm with competitive ratio better than $2$ must assign these jobs so that the makespan equals $1$ at the end; otherwise, we may stop revealing further jobs so that the offline optimum equals $1$, while the makespan equals $2$.

The fifth job fulfills $5\in S_1\setminus S_2$ and $p_5=X$. Without loss of generality, it is assigned to the first machine.

The sixth job is given by $6\in S_2\setminus S_1$, $p_6=\frac X 2$. Now there are two options:
\begin{enumerate}[(i)]
    \item If the sixth job is also assigned to the first machine, we reveal the seventh job with $7\in S_2\setminus S_1$ and $p_7=X$. This job cannot be assigned to the first machine, as the makespan would then be $1+\frac{3X}{2}$, while the offline optimum is $X$. Indeed, this would imply a ratio of at least
    \[\frac{1+\frac{9+3\sqrt{17}}{2}}{3+\sqrt{17}}=\frac{9+\sqrt{17}}{8}.\]
    In particular, we obtain $p(S_1\cap J_1\cap [7])=p(S_2\cap J_2\cap [7])=1+X$.
    \item If the sixth job is assigned to the second machine, we reveal the seventh job as $7\in S_1\cap S_2$ and $p_7=\frac X 2$. Due to similar reasoning to the other case, the seventh job must be assigned to the second machine and we again have $p(S_1\cap J_1\cap [7])=p(S_2\cap J_2\cap [7])=1+X$.
\end{enumerate}
See Figure \ref{fig:53lb} for the first seven jobs in each of the cases.

In both of these cases, we reveal the eighth job with $8\in S_1\cap S_2$ and $p_8=2+\frac{3X}{2}$. No matter where this job is assigned, the makespan is $3+\frac{5X}{2}$, whereas the optimal solution $\{1,\ldots, 7\}\dot\cup\{8\}$ achieves a makespan of $2+\frac{3X}{2}$. Therefore, the competitive ratio is bound to be at least
\[\frac{3+\frac{5X}{2}}{2+\frac{3X}{2}}=\frac{6+5(3+\sqrt{17})}{4+3(3+\sqrt{17})}=\frac{9+\sqrt{17}}{8},\]
as desired.
\end{proof}
\begin{figure}
    \centering
    \begin{tikzpicture}
             \draw [fill=blue!30] (0,0) rectangle (0.5,0.5) node[pos=.5] {$4$};
              \draw [fill=blue!30] (0,0.5) rectangle (0.5,1) node[pos=.5] {$3$};
            \draw [fill=blue!30] (0,1.2) rectangle (0.5,1.7) node[pos=.5] {$2$};
             \draw [fill=blue!30] (0,1.7) rectangle (0.5,2.2) node[pos=.5] {$1$};
                \draw [fill=blue!30] (0.5,1.7) rectangle (4,2.2) node[pos=.5] {$5$};
                \draw [fill=blue!30] (0.5,1.2) rectangle (2.25,1.7) node[pos=.5] {$6$};
                  \draw [fill=blue!30] (0.5,0) rectangle (4,0.5) node[pos=.5] {$7$};
                
                        \draw [fill=blue!30] (6,0) rectangle (6.5,0.5) node[pos=.5] {$4$};
              \draw [fill=blue!30] (6,0.5) rectangle (6.5,1) node[pos=.5] {$3$};
            \draw [fill=blue!30] (6,1.2) rectangle (6.5,1.7) node[pos=.5] {$2$};
                 \draw [fill=blue!30] (6,1.7) rectangle (6.5,2.2) node[pos=.5] {$1$};
                \draw [fill=blue!30] (6.5,1.7) rectangle (10,2.2) node[pos=.5] {$5$};
                \draw [fill=blue!30] (6.5,0) rectangle (8.25,0.5) node[pos=.5] {$6$};
                 \draw [fill=blue!30] (6.5,0.5) rectangle (8.25,1) node[pos=.5] {$7$};
                  \draw [fill=blue!30] (8.25,0) rectangle (10,0.5) node[pos=.5] {$7$};

              \node at(-1, 0.2){$J_2\cap S_2$};
              \node at(-1, 0.7){$J_2\cap S_1$};
              \node at(-1, 1.4){$J_1\cap S_2$};
              \node at(-1, 1.9){$J_1\cap S_1$};
     
\end{tikzpicture}

    \caption{The first seven jobs in cases (i) and (ii) of the proof of Theorem \ref{thm:53lb}, respectively. In both cases, an eighth job $8\in S_1\cap S_2$ forces the schedules to the desired lower bound of competitiveness.}
    \label{fig:53lb}
\end{figure}
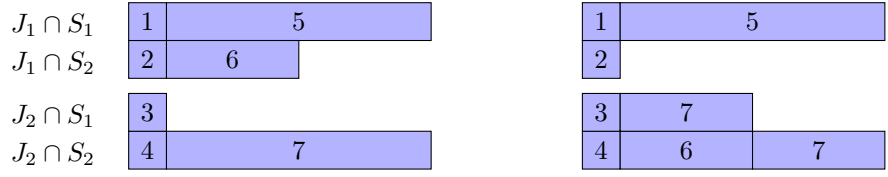


One might wonder how the gap between the upper bound in Theorem \ref{thm:53ub} and the lower bound in Theorem \ref{thm:53lb} can be narrowed down. We can observe through a simple family of instances that, if that is possible at all by means of improving the upper bound, then dropping Rule \ref{asp:rulefix} is necessary.

\begin{obs}\label{obs:lbobs}
    Let $r<\nicefrac{5}{3}$ be arbitrary. There exists no deterministic $r$-competitive algorithm for \textsc{OMSS(2,2)} which satisfies Rule \ref{asp:rulefix}.
\end{obs}
\begin{proof}
   Consider an instance of $n=5$ jobs with processing times $p_1=p_2=p_4=1$, $p_3=2$, $p_5=3$. The scenarios are given by $S_1=\{1,2,4,5\}$, $S_2=\{3,4,5\}$.

   The first two jobs must be assigned to opposite machines; for otherwise one could \emph{truncate} the instance, i.e., consider the restricted instance of the first $n=2$ jobs with $p_j\equiv 1$, $S_1=\{1,2\}$, $S_2=\emptyset$ instead. On this restricted instance, the algorithm would be no better than $2$-competitive. Hence, we may assume without loss of generality that $1\in J_1$ and $2\in J_2$.

   The third job is assigned to the first machine without loss of generality. Then, as $4\in S_1\cap S_2$ and the first machine uniquely admits makespan, we must assign the fourth job to the second machine, and similarly, the fifth job is assigned to the second machine as well.

   In total, the makespan is given by
   \[p(J_1\cap S_2)=1+1+3=5,\]
    although an optimal solution is given by $J_1=\{1,2,3,4\}$ and $J_2=\{5\}$, which has a makespan of $3$.
\end{proof}

\begin{figure}
    \centering
    \begin{tikzpicture}
             \draw [fill=blue!30] (0,0) rectangle (1,0.5) node[pos=.5] {$1$};
            \draw [fill=blue!30] (0,1.2) rectangle (1,1.7) node[pos=.5] {$2$};
       \draw [fill=blue!30] (0,0.5) rectangle (1,1) node[pos=.5] {$4$};
       \draw [fill=blue!30] (1,0.5) rectangle (4,1) node[pos=.5] {$5$};
        \draw [fill=blue!30] (1,0) rectangle (2,0.5) node[pos=.5] {$4$};
             \draw [fill=blue!30] (2,0) rectangle (5,0.5) node[pos=.5] {$5$};
             \draw [fill=blue!30] (0,1.7) rectangle (2,2.2) node[pos=.5] {$3$};

    \draw [fill=blue!30] (7,1.7) rectangle (9,2.2) node[pos=.5] {$3$};
     \draw [fill=blue!30] (9,1.7) rectangle (10,2.2) node[pos=.5] {$4$};
      \draw [fill=blue!30] (7,1.2) rectangle (8,1.7) node[pos=.5] {$1$};
        \draw [fill=blue!30] (8,1.2) rectangle (9,1.7) node[pos=.5] {$2$};
          \draw [fill=blue!30] (9,1.2) rectangle (10,1.7) node[pos=.5] {$4$};
               \draw [fill=blue!30] (7,0) rectangle (10,0.5) node[pos=.5] {$5$};
                     \draw [fill=blue!30] (7,0.5) rectangle (10,1) node[pos=.5] {$5$};

              \node at(-1, 0.2){$J_2\cap S_2$};
              \node at(-1, 0.7){$J_2\cap S_1$};
              \node at(-1, 1.4){$J_1\cap S_2$};
              \node at(-1, 1.9){$J_1\cap S_1$};
     
\end{tikzpicture}

    \caption{The online schedule in the proof of Observation \ref{obs:lbobs} and the offline optimum.}
    \label{fig:53lb2}
\end{figure}
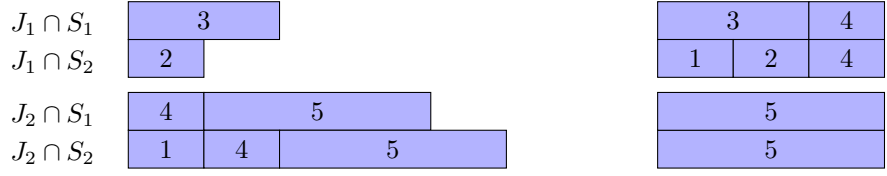

In the above setting, the right course of action would have indeed been to violate Rule \ref{asp:rulefix} and assign $4\in J_1$. Hence, it is justified to suspect that the algorithm can be improved by altering this rule according to the sizes of jobs $j\in S_1\cap S_2$, although as seen in Theorem \ref{thm:53lb}, it can only improve the ratio by less than $0.027$.

\section{Tight Bounds for $m=3$ and Sufficiently Many Scenarios}\label{sec:hypergraphs}

One of our main results (Theorem \ref{thm:general}) states that for any number $m\in \mathbb N$ of machines, there exists a number $K$ of scenarios such that no deterministic algorithm  for \textsc{OMSS$(m,K)$} can have a competitive ratio less than $m$, even for unit processing times. 

\begin{theorem}\label{thm:nobetterthan3}
    Let $K\geq 233$. Then for \textsc{Online Makespan Hypergraph Coloring}$(m)$, there exists no $r$-competitive online algorithm for any $r<3$, even when the problem is restricted to $3$-uniform instances with at most $K$ hyperedges.
\end{theorem}

The problem \textsc{Online Makespan Hypergraph Coloring}$(m)$, which we formulated in Section \ref{sec:prelim}, is equivalent to the unit processing time case of \textsc{OMSS$(m,K)$} with $K$ part of the input, so that the lower bounds are immediately implied for the latter as well for $K\geq 233$.

 To prove Theorem \ref{thm:nobetterthan3}, we provide a family of instances on which any online algorithm is forced to output a monochromatic hyperedge of size $3$. For each of the instances, we also maintain (offline) colorings $c\colon [j]\to \{\offlinecolors\}$ to be \emph{feasible}, i.e., with the property that if two nodes $v,v'$ are included in a common hyperedge $e$, we have $c(v)\neq c(v')$. The latter coloring suggests an offline optimum of value $1$ for each of the instances, implying the competitive ratio of $3$ immediately. The main difference between the online and offline colorings is, as expected, that we may change the offline coloring completely at all times as long as it remains feasible. Among other modifications, we often employ a permutation 
$\pi\colon \{\mathrm{red, blue, yellow}\}\to \{\mathrm{red, blue, yellow}\}$
to realize this possibility. We must trace the offline coloring simultaneously because the hyperedges are often defined adversarially according to both colorings of the current hypergraph.  

In the following, let $J_i\subseteq V$ denote the subset of nodes assigned to the online color $i$ for $i\in\{1,2,3\}$. Right before the insertion of a node $j$, we implicitly mean $J_i\cap [j-1]$ whenever we mention $J_i$ for some $i\in \{1,2,3\}$. Inspired by our original scheduling problem, we call the objective value $\max_{e\in E}\max_{i\in \{1,2,3\}}|e\cap J_i|$ of an instance its \emph{makespan}. Without loss of generality, we can truncate an instance as soon as it obtains a makespan of $3$.

\subsection{Our gadget: Seven-node subgraph $\mathcal S$}
Initially, we study a subgraph $\mathcal S$ that we use as a building block throughout our construction. The hypergraph $\mathcal S$ is given by the nodes $V(\mathcal S)=\{1,\ldots, 7\}$ and hyperedges 
\begin{align*}
    E(\mathcal S)=\big\{&\{1,2\},\{1,3\},\{2,3\},\{1,2,3\},\{2,4\},\{4,5\},\{2,5\},\{2,4,5\},\\&\{4,6\}, \{6,7\},\{4,7\}, \{4,6,7\},\{3,6\}, \{3,7\}\big\}.
\end{align*}
One can verify that the hyperedges have size at most $3$ and the hyperedge set $E(\mathcal S)$ is closed under taking subsets of size at least $2$, i.e., if $e\in E(\mathcal S)$, $|e|=3$ and $e'\subseteq e$ with $|e'|=2$, then $e'\in E(\mathcal S)$. The hypergraph $\mathcal S$ together with the maximal hyperedges $e_1,\ldots, e_5$ is given in Figure \ref{fig:hypergraphs}.

  \begin{figure}[h!]
    \centering
    \begin{tikzpicture}
     \coordinate (C) at (0.95,0.45);
      \coordinate(A) at (0.45,1.45);
     \coordinate(B) at (1.45,1.45);
     \coordinate(E) at (1.95,0.45);
     \coordinate(D) at (2.45,1.45);
     \coordinate(F) at (2.95,0.45);
         \coordinate(G) at (3.95,0.45);

         \draw[fill=gray!30, thick] (A) -- (B) -- (C) -- cycle;
\node at (barycentric cs:B=1,A=1,C=1) {$e_1$};
  \draw[fill=gray!30, thick] (B) -- (D) -- (E) -- cycle;
\node at (barycentric cs:D=1,B=1,E=1) {$e_2$};
  \draw[fill=gray!30, thick] (G) -- (D) -- (F) -- cycle;
\node at (barycentric cs:G=1,D=0.7,F=0.7) {$e_3$};
    \draw[thick] (G) to [bend left=70] node[below]{$e_4$} (C) ;
     \draw[thick] (F) to [bend left=50] node[below right]{$e_5$} (C) ;
     \draw  node [neutral](3) at (0.95,0.45) {$3$};
     \draw  node [neutral](1) at (0.45,1.45) {$1$};
     \draw  node [neutral](2) at (1.45,1.45) {$2$};
     \draw  node [neutral](5) at (1.95,0.45) {$5$};
     \draw  node [neutral](4) at (2.45,1.45) {$4$};
      \draw  node [neutral](6) at (3.95,0.45) {$6$};
         \draw  node [neutral](7) at (2.95,0.45) {$7$};
\end{tikzpicture}
    \caption{The hypergraph $\mathcal S$ with the maximal hyperedges $e_1,\ldots, e_5$.}
    \label{fig:hypergraphs}
\end{figure}
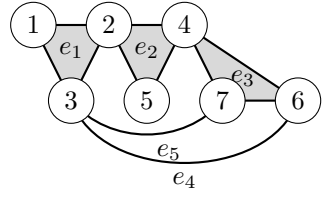

We first assume that the seven nodes are assigned to exactly two out of three online colors. 
Since the online colors are initially symmetric, let us assume that we use the online colors $1$ and $2$.

\begin{asp}\label{asp:twomachs}
$e_{k}\cap J_3=\emptyset$ for all $k\in \{1,2,3\}$.
\end{asp}

Moreover, for brevity, we only mention the revelation of inclusion-wise maximal hyperedges whenever we reveal a copy of this gadget, although the subsets of these hyperedges of size $2$ are being introduced as well. 

First, we describe a case distinction about the introduction of the gadget, always satisfying Assumption \ref{asp:twomachs}. We reveal the first three nodes that appear in the common hyperedge $e_{1}=\{1,2,3\}$. If all three nodes were assigned to a common online color, we would obtain a makespan of $3$, which ultimately proves the theorem. Therefore, we assume that the nodes were partitioned in two online colors; as the nodes are symmetric so far, we may assume without loss of generality that the first two nodes were assigned to the first color and the third job to the second, see Figure \ref{fig:firstthree}. 

\begin{figure}[h!]
    \centering
    \begin{tikzpicture}
    \coordinate (A) at (0.45,0.45);
    \coordinate (B) at (0.45,1.45);
    \coordinate (C) at (1.45,1.45);
    \draw[fill=gray!30, thick] (A) -- (B) -- (C) -- cycle;
\node at (barycentric cs:A=1,B=0.7,C=1) {$e_1$};

      \draw  node [myred] (3) at (0.45,0.45) {$3$};
       \draw  node [myblue] (1) at (0.45,1.45) {$1$};
      \draw  node [myyellow] (2) at (1.45,1.45) {$2$};
     
        \draw  node [] at (-1.45,1.45) {$J_1$};
       \draw  node [] at (-1.45,0.45) {$J_2$};
   \end{tikzpicture}
    \caption{First three nodes.}
    \label{fig:firstthree}
\end{figure}
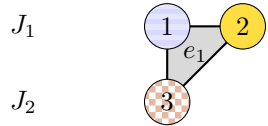

Now we introduce two more nodes $4,5$ in this order one by one and reveal that 
\[e_{2}=\{2,4,5\}.\]

If both nodes $4$ and $5$ are assigned to the second online color, we may assign the remaining two nodes $6$ and $7$ arbtirarily and ignore them from now on, see Figure \ref{fig:allfive}. 
    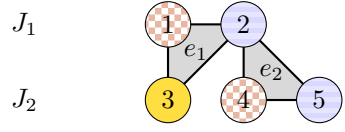
\begin{figure}[h!]
    \centering
    \begin{tikzpicture}
         \coordinate (C) at (0.45,0.45);
      \coordinate(A) at (0.45,1.45);
     \coordinate(B) at (1.45,1.45);
     \coordinate(D) at (1.45,0.45);
     \coordinate(E) at (2.45,0.45);
        \draw[fill=gray!30, thick] (A) -- (B) -- (C) -- cycle;
\node at (barycentric cs:B=1,A=0.7,C=1) {$e_1$};

   \draw[fill=gray!30, thick] (B) -- (D) -- (E) -- cycle;
\node at (barycentric cs:B=1,D=0.7,E=1) {$e_2$};
    
     \draw  node [myyellow] (3) at (0.45,0.45) {$3$};
      \draw  node [myred](1) at (0.45,1.45) {$1$};
     \draw  node [myblue](2) at (1.45,1.45) {$2$};
     \draw  node [myred](4) at (1.45,0.45) {$4$};
     \draw  node [myblue](5) at (2.45,0.45) {$5$};
     
        \draw  node [] at (-1.45,1.45) {$J_1$};
       \draw  node [] at (-1.45,0.45) {$J_2$};
\end{tikzpicture}
    \caption{The five relevant nodes if $\tau(4)=\tau(5)=2$.}
    \label{fig:allfive}
\end{figure}

Otherwise, one of the nodes $4$ and $5$, without loss of generality, the node $4$, is assigned to the first color. We reveal the sixth node along with the hyperedges
\[e_3\supseteq\{4,6\}, e_4\supseteq\{3,6\}.\]
If the sixth node is assigned to the second color, we may ignore the seventh node similar to above, see Figure \ref{fig:allsix}.
    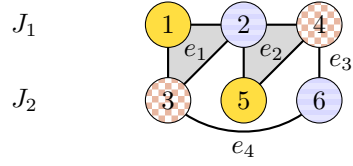
\begin{figure}[h!]
    \centering
    \begin{tikzpicture}
      \coordinate (C) at (0.45,0.45);
      \coordinate(A) at (0.45,1.45);
     \coordinate(B) at (1.45,1.45);
     \coordinate(E) at (1.45,0.45);
     \coordinate(D) at (2.45,1.45);
     \coordinate(F) at (2.45,0.45);
        \draw[fill=gray!30, thick] (A) -- (B) -- (C) -- cycle;
\node at (barycentric cs:B=1,A=0.7,C=1) {$e_1$};
  \draw[fill=gray!30, thick] (B) -- (D) -- (E) -- cycle;
\node at (barycentric cs:D=1,B=0.7,E=1) {$e_2$};
    \draw[thick](F)--(D) node [midway, right]{$e_3$};
    \draw[thick] (F) to [bend left=45] node[below]{$e_4$} (C) ;

     \draw  node [myred](3) at (0.45,0.45) {$3$};
     \draw  node [myyellow](1) at (0.45,1.45) {$1$};
    \draw  node [myblue](2) at (1.45,1.45) {$2$};
      \draw  node [myyellow](5) at (1.45,0.45) {$5$};
    \draw  node [myred](4) at (2.45,1.45) {$4$};
     \draw  node [myblue](6) at (2.45,0.45) {$6$};
     
        \draw  node [] at (-1.45,1.45) {$J_1$};
       \draw  node [] at (-1.45,0.45) {$J_2$};
\end{tikzpicture}
    \caption{The six relevant nodes if $\tau(6)=2$.}
    \label{fig:allsix}
\end{figure}

If the sixth node is assigned to the first online color, we reveal the seventh node along with the hyperedges 
\[e_3=\{4,6,7\}, e_5\supseteq \{3,7\}.\]

Throughout the remainder of this section, we let $\mathcal S$ denote the resulting hypergraph with seven nodes and fourteen hyperedges (including subsets).

Since both the sixth and the seventh nodes are assigned to the first online color, the seventh node is forced to be assigned to the second online color. Our final assignment looks as in Figure \ref{fig:allseven}.

    \begin{figure}[h!]
    \centering
    \begin{tikzpicture}
         \coordinate(G) at (3.45,1.45);

         \draw[fill=gray!30, thick] (A) -- (B) -- (C) -- cycle;
\node at (barycentric cs:B=1,A=0.7,C=1) {$e_1$};
  \draw[fill=gray!30, thick] (B) -- (D) -- (E) -- cycle;
\node at (barycentric cs:D=1,B=0.7,E=1) {$e_2$};
  \draw[fill=gray!30, thick] (G) -- (D) -- (F) -- cycle;
\node at (barycentric cs:G=1,D=0.7,F=1) {$e_3$};
    \draw[thick] (G) to [out=340, in=270, looseness=2] node[below]{$e_4$} (C) ;
     \draw[thick] (F) to [bend left=50] node[below]{$e_5$} (C) ;
     \draw  node [myred](3) at (0.45,0.45) {$3$};
     \draw  node [myyellow](1) at (0.45,1.45) {$1$};
     \draw  node [myblue](2) at (1.45,1.45) {$2$};
     \draw  node [myyellow](5) at (1.45,0.45) {$5$};
     \draw  node [myred](4) at (2.45,1.45) {$4$};
      \draw  node [myyellow](6) at (3.45,1.45) {$6$};
         \draw  node [myblue](7) at (2.45,0.45) {$7$};
         \draw  node [] at (-1.45,1.45) {$J_1$};
       \draw  node [] at (-1.45,0.45) {$J_2$};
\end{tikzpicture}
    \caption{The hypergraph $\mathcal S$ together with the online and offline colorings, under the assumption that $\tau(6)=1$.}
    \label{fig:allseven}
\end{figure}
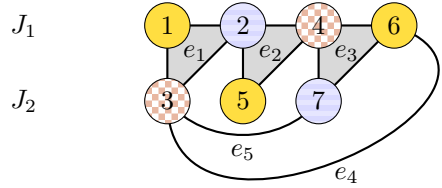

\begin{prop}\label{prop:bicolorededges}
    For any partitioning $V(\mathcal S)=J_1\dot\cup J_2$ of the nodes that admits a makespan of at most $2$, there exists a node coloring $c\colon V(\mathcal S)\to \{\offlinecolors\}$ such that there are two hyperedges $e\in E(\mathcal S[J_1]), e'\in E(\mathcal S[J_2])$ with $c(e)=c(e')=\{\mathrm{red},\mathrm{blue}\}$.
\end{prop}
\begin{proof}
    See Figures \ref{fig:allfive}, \ref{fig:allsix} and \ref{fig:allseven} for the offline coloring in the respective cases. The hyperedges $(e,e')$ in question are $(\{1,2\},\{4,5\})$; $(\{2,4\}, \{3,6\})$ and $(\{2,4\},\{3,7\})$, respectively.
\end{proof}

In other words, we have induced monochromatic hyperedges of both colors with respect to the online coloring. Applying this repeatedly, we can force the algorithm to a monochromatic hyperedge of size $3$, which we show next.

\subsection{Pigeonholing the two-color case}\label{subs:pigeon7}

In our construction, we would like to create \emph{copies} of the graph $\mathcal S$ successively, which we denote by $\mathcal S(t)$ for incrementing $t\in \mathbb N$. 

A straightforward application of the pigeonhole principle yields that, if we have seven subgraphs $\mathcal S(1),\ldots,\mathcal S(7)$ that are each partitioned onto two online colors, at least three of them (without loss of generality $\mathcal S(1)$, $\mathcal S(2)$, $\mathcal S(3)$) must be partitioned onto the same two online colors. Again without loss of generality, these two online colors are $1$ and $2$. Considering Proposition \ref{prop:bicolorededges}, we can color the three copies $\mathcal S(1)$, $\mathcal S(2)$ and $\mathcal S(3)$ so as to satisfy the following properties:

\begin{itemize}
    \item There exist two hyperedges $e_i^1=\{j_i^1, j_i^2\}\in  E(\mathcal S(1))$ for $i\in \{1,2\}$ with the property that $c(e_i^1)=\{\mathrm{red,blue}\}$,
     \item there exist two hyperedges $e_i^2=\{j_i^3, j_i^4\}\in  E(\mathcal S(2))$ for $i\in \{1,2\}$ with the property that $c(e_i^1)=\{\mathrm{blue, yellow}\}$,
    \item there exist two hyperedges $e_i^3=\{j_i^5, j_i^6\}\in  E(\mathcal S(3))$ for $i\in \{1,2\}$ with the property that $c(e_i^3)=\{\mathrm{yellow, red}\}$.
\end{itemize}

Moving forward, we restrict our subgraph to one with the nodes $j_i^t$ for $i\in \{1,2\}, t\in \{1,\ldots,6\}$ and the hyperedges $k_i^t$ for $i\in \{1,2\}, t\in \{1,2,3\}$. The remaining nodes in $V(\mathcal S(1))\dot\cup V(\mathcal S(2))\dot \cup V(\mathcal S(3))$ will not be connected to other nodes anymore, nor will their offline coloring be altered. Thus, they can be ignored from now on. 

    \begin{figure}[h!]
    \centering
    \begin{tikzpicture}
     
   \draw  node [myred](2) at (0.45,0.45) {$j_2^1$};
         \draw  node [myred](1) at (0.45,1.45) {$j_1^1$};
     \draw  node [myblue](4) at (1.45,0.45) {$j_2^2$};
     \draw  node [myblue](3) at (1.45,1.45) {$j_1^2$};
      \draw[thick](3)--(1) node [midway, above]{$e_1^1$};
  \draw[thick](4)--(2) node [midway, above]{$e_2^1$};
 
      \draw  node [myblue](6) at (3.45,0.45) {$j_2^3$};
      \draw  node [myblue](5) at (3.45,1.45) {$j_1^3$};
      \draw  node [myyellow](8) at (4.45,0.45) {$j_2^4$};
      \draw  node [myyellow](7) at (4.45,1.45) {$j_1^4$};
           \draw[thick](5)--(7) node [midway, above]{$e_1^2$};
  \draw[thick](6)--(8) node [midway, above]{$e_2^2$};

          \draw  node [myyellow](10) at (6.45,0.45) {$j_2^5$};
       \draw  node [myyellow](9) at (6.45,1.45) {$j_1^5$};
       \draw  node [myred](12) at (7.45,0.45) {$j_2^6$};
       \draw  node [myred](11) at (7.45,1.45) {$j_1^6$};
          \draw[thick](9)--(11) node [midway, above]{$e_1^3$};
  \draw[thick](10)--(12) node [midway, above]{$e_2^3$};

        \draw  node [] at (-1.45,1.45) {$J_1$};
       \draw  node [] at (-1.45,0.45) {$J_2$};
\end{tikzpicture}
    \caption{The nodes $j_i^t$.}
    \label{fig:sevenpigeon}
\end{figure}
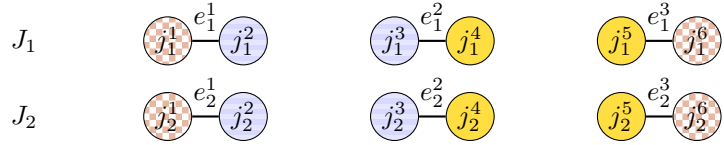

Now we introduce the thirteenth node $x_1$, revealing the full extent of the hyperedges $e_1^1=\{j_1^1,j_1^2,x_1\}$ and $e_2^1=\{j_2^1,j_2^2,x_1\}$. Unless a competitive ratio of $3$ is attained, the node $x_1$ must be assigned to the third online color. Thereafter, the fourteenth node $x_2$ is revealed together with hyperedges $e_1^2=\{j_1^3,j_1^4,x_2\}$, $e_2^2=\{j_2^3,j_2^4,x_2\}$, $f_1\supseteq \{x_1,x_2\}$. Analogously, the fourteenth node is assigned to the third online color as well. Notice that we can feasibly extend the offline coloring by $c(x_1)=\mathrm{yellow}$ and $c(x_2)=\mathrm{red}$, as depicted in Figure \ref{fig:14jobs}.

Finally, we would like to introduce the last node $x_3$ with the hyperedges $e_1^3=\{j_1^5,j_1^6,x_3\}$, $e_2^3=\{j_2^5,j_2^6,x_3\}$ and $f=\{x_1,x_2,x_3\}$. For each online color, there exists a hyperedge $e\in \{e_1^3,e_2^3, f\}$ with $x_3\in e$ that already attains a makespan of $2$. Therefore, the node $x_3$ cannot be placed without increasing the makespan to $3$. On the other hand, extending the coloring $c$ by $c(v_n)=\mathrm{blue}$, as suggested by Figure \ref{fig:14jobs}, yields a feasible coloring of the nodes, suggesting that the offline partitioning given by the colors attains a makespan of $1$ in every scenario.

As a consequence, we obtain:

\begin{lemma}\label{lemma:twocolor}
    Any algorithm that assigns at least $7$ copies of $\mathcal S$ to at most two online colors each is no better than $3$-competitive.
\end{lemma}
 \begin{figure}
    \centering
    \begin{tikzpicture}
    \clip (-3,-1.5) rectangle + (12.2,3.9);
    \draw[fill=gray!30, thick] (3.45,0.45) -- (4.45,0.45) -- (3.45,-0.55) --cycle;
        \draw[fill=gray!30, thick] (0.45,0.45) -- (1.45,0.45) -- (0.45,-0.55) --cycle;

                \draw[fill=gray!30, thick] (6.45,0.45) -- (7.45,0.45) -- (6.45,-0.55) --cycle;
       \draw[draw=black, thick, fill=gray!30] (0.45,1.45) to [out=0, in=180] (1.45,1.45) .. controls (0,3.5) and (-1.5,1) .. node[below right]{$e_1^1$}(0.45,-0.55) to [bend left=40](0.45,1.45);
       
        \draw[draw=black, thick, fill=gray!30] (3.45,1.45) to [out=0, in=180] (4.45,1.45) .. controls (3,3.5) and (1.5,1) .. node[below right]{$e_1^2$}(3.45,-0.55) to [bend left=40](3.45,1.45);
        \draw[draw=black, thick, fill=gray!30] (6.45,1.45) to [out=0, in=180] (7.45,1.45) .. controls (6,3.5) and (4.5,1) .. node[below right]{$e_1^3$}(6.45,-0.55) to [bend left=40](6.45,1.45);
         \draw[draw=black, thick, fill=gray!30] (0.45,-0.55) to [out=0, in=180] (3.45,-0.55) to (6.45,-0.55) .. controls (4.45,-1.5) and (2.45,-1.5) .. node[above right]{$f$}(0.45,-0.55);
        \draw  node [myred](2) at (0.45,0.45) {$j_2^1$};
         \draw  node [myred](1) at (0.45,1.45) {$j_1^1$};
     \draw  node [myblue](4) at (1.45,0.45) {$j_2^2$};
     \draw  node [myblue](3) at (1.45,1.45) {$j_1^2$};
 
      \draw  node [myblue](6) at (3.45,0.45) {$j_2^3$};
      \draw  node [myblue](5) at (3.45,1.45) {$j_1^3$};
      \draw  node [myyellow](8) at (4.45,0.45) {$j_2^4$};
      \draw  node [myyellow](7) at (4.45,1.45) {$j_1^4$};

          \draw  node [myyellow](10) at (6.45,0.45) {$j_2^5$};
       \draw  node [myyellow](9) at (6.45,1.45) {$j_1^5$};
       \draw  node [myred](12) at (7.45,0.45) {$j_2^6$};
       \draw  node [myred](11) at (7.45,1.45) {$j_1^6$};

       \draw  node [myyellow](13) at (0.45,-0.55) {$x_1$};
       \draw  node [myred](14) at (3.45,-0.55) {$x_2$};
        \draw  node [myblue](15) at (6.45,-0.55) {$v_n$};
       \draw  node [] at (-2.45,1.45) {$J_1$};
       \draw  node [] at (-2.45,0.45) {$J_2$};      \draw  node [] at (-2.45,-0.55) {$J_3$};
  
\node at (barycentric cs:4=1,2=0.7,13=1) {$e_2^1$};
\node at (barycentric cs:8=1,6=0.7,14=1) {$e_2^2$};
\node at (barycentric cs:12=1,10=0.7,15=1) {$e_2^3$};
\end{tikzpicture}
    \caption{All fifteen relevant nodes together with the  hyperedges within them. Notice that any assignment of the node $v_n$ leads to a makespan of $3$.} 
    \label{fig:14jobs}
\end{figure}
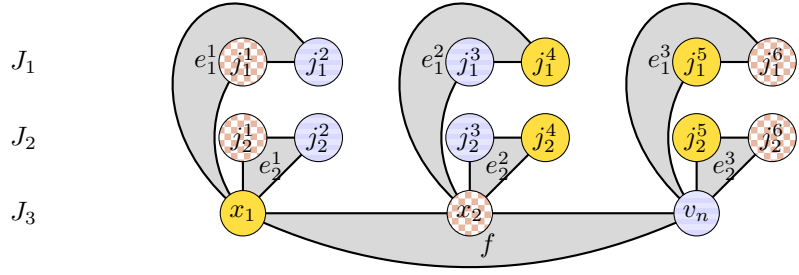

\subsection{Subgraph $\mathcal S$ on Three Online Colors: The Palettes}
  It remains to consider the case where copies of the subhypergraph $\mathcal S$ are partitioned among the three colors, i.e., its nodes $\{1,\ldots,7\}$ are partitioned in nonempty sets $ J_1\dot\cup  J_2\dot\cup  J_3$. 

\begin{lemma}\label{lemma:paletteexists}
    For every online coloring $J_1\dot\cup J_2\dot \cup J_3$ of the nodes $V(\mathcal S)=\{v_1,\ldots,v_7\}$ for nonempty sets $ J_1\dot\cup  J_2\dot\cup  J_3$, there is an offline coloring $c\colon V(\mathcal S)\to \{\mathrm{red}, \mathrm{blue},\mathrm{yellow}\}$ such that two nodes exist with different offline colors but the same online color. 
\end{lemma}
\begin{proof}
    We fix the offline coloring in Figure \ref{fig:allseven}, i.e., nodes $3,4$ are red, nodes $2,7$ are blue and nodes $1,5,6$ are yellow. The only case where the claim does not hold for the coloring is when each $J_i$ consists of nodes of precisely one offline color. In other words, up to a permutation of the colors, we have $J_1=c^{-1}(\{\mathrm{red}\})$, $J_2=c^{-1}(\{\mathrm{blue}\})$, $J_3=c^{-1}(\{\mathrm{yellow}\})$. In this case, we alter the offline coloring as follows: $1$ and $5$ are blue instead of yellow, and $2$ is yellow instead of blue. It is straightforward to check that this is a feasible coloring, see Figure \ref{fig:proofpalette}.
        \begin{figure}
    \centering
    \begin{tikzpicture}
          \coordinate (C) at (0.45,2.45);
      \coordinate(A) at (0.45,0.45);
     \coordinate(B) at (0.45,1.45) ;
     \coordinate(E) at  (1.45,0.45);
     \coordinate(D) at (1.45,2.45);
     \coordinate(F) at (2.45,0.45);
         \coordinate(G) at  (1.45,1.45);

         \draw[fill=gray!30, thick] (A) to (B) to (C) to[bend right=50](A) -- cycle;
\node at (barycentric cs:B=1,A=0.7,C=1) {};
  \draw[fill=gray!30, thick] (B) to (D) to [bend right=50](E) -- cycle;
\node at (barycentric cs:D=1,B=0.7,E=1) {};
  \draw[fill=gray!30, thick] (G) -- (D) -- (F) -- cycle;
\node at (barycentric cs:G=1,D=0.7,F=1) {};
    \draw[thick] (G) to node[below]{} (C) ;
     \draw[thick] (F) to [out=60, in=60, looseness=2] node[below]{} (C) ;
     \draw  node [myred](3) at (0.45,2.45) {$3$};
     \draw  node [myblue](1) at (0.45,0.45) {$1$};
     \draw  node [myyellow](2) at (0.45,1.45) {$2$};
     \draw  node [myblue](5) at (1.45,0.45) {$5$};
     \draw  node [myred](4) at (1.45,2.45) {$4$};
      \draw  node [myyellow](6) at (2.45,0.45) {$6$};
         \draw  node [myblue](7) at (1.45,1.45) {$7$};
\end{tikzpicture}
    \caption{The second case in the proof of Lemma \ref{lemma:paletteexists}. Names of hyperedges are left out for brevity.}
    \label{fig:proofpalette}
\end{figure}
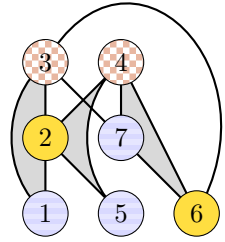
\end{proof}

 In our construction, we exploit this property as follows: Let us consider a copy $\mathcal S(t)$ and fix two of its  nodes $v_i(t)\in J_i$, $v_{i'}(t)\in J_{i'}$ provided by Lemma \ref{lemma:paletteexists} as well as a third node $v_{i''}(t)\in J_{i''}$, where $i''\notin \{i,i'\}$. By their choice, the nodes $v_i(t)$ ($i\in \{1,2,3\}$) admit at most two offline colors and avoid a third offline color $C\in \{\mathrm{blue},\mathrm{red},\mathrm{yellow}\}$. 
 
In this case, the subhypergraph $\mathcal S(t)$ is called a \emph{$ C$ palette} (i.e., \emph{yellow palette}, \emph{blue palette} or \emph{red palette}) and the corresponding nodes $v_i(t)$ are called \emph{palette nodes}. In particular, $c(v_i(t))\neq C$. The choice of the name is motivated by the fact that, if we connect a new node $w$ to e.g.~yellow palette nodes, we are still allowed to extend the offline coloring by $c(w)=\mathrm{yellow}$, while maintaining that the makespan of the offline coloring is $1$.

 \subsection{The Construction}
 
 Below is a summary of our construction for the case where Lemma \ref{lemma:twocolor} is not applicable. Essentially, we create copies of the subgraph $\mathcal S$ which we consecutively connect with the existing hypergraph in varying ways.
 Below is a summary of our construction for the case where Lemma \ref{lemma:twocolor} is not applicable. Essentially, we create copies of the subhypergraph $\mathcal S$ which we connect with the existing hypergraph. The online color of the node $v$ is denoted by $\varphi(v)$ and its (current) offline color by $c(v)$.
 \begin{enumerate}[(i)]
     \item \label{item:one}We create up to $13$ copies of the subhypergraph $\mathcal S$. By Lemma \ref{lemma:twocolor}, we may assume that at least $7$ of the copies are assigned to all three online colors (we truncate this step as soon as $7$ such copies can be found). We handpick $7$ such copies of $\mathcal S$, calling them $\mathcal S(1),\ldots, \mathcal S(7)$. 

     \item \label{item:two}We permute the offline coloring on each copy of $\mathcal S$ so that $\mathcal S(1)$, $\mathcal S(4)$ and $\mathcal S(7)$ are yellow palettes,
 $\mathcal S(2)$ and $\mathcal S(5)$ are blue palettes,
and $\mathcal S(3)$ and $\mathcal S(6)$ are red palettes.

     \item \label{item:three}We reveal a node $w_\mathcal A$ which is connected to palette nodes $v_i(1)\in V(\mathcal S(1))\cap J_i$ for $i\in \{1,2,3\}$. The assignment to the online color $\varphi(w_\mathcal A)$ leads to a monochromatic hyperedge $\{v_{\varphi(w_\mathcal A)}(1), w_\mathcal A\}$ of online color $\varphi(w_\mathcal A)$. We define $v_\mathcal A\coloneqq v_{\varphi(w_\mathcal A)}(1)$.

     \item \label{item:four} We reveal three nodes $w_{\mathcal B1}$,
$w_{\mathcal B2}$ and $w_{\mathcal B3}$ that are connected by the hyperedge $\{w_{\mathcal B1},w_{\mathcal B2},w_{\mathcal B3}\}$ to each other. Moreover, for each $s\in \{1,2,3\}$, the node $w_{\mathcal Bs}$ is connected to the palette nodes $v_i(s+1)$ for $i\in \{1,2,3\}$. 
There is an $i\neq \varphi(w_\mathcal A)$ such that $\{v_i(s+1), w_{\mathcal Bs}\}$ is a monochromatic edge of online color $i$ for a suitable copy $\mathcal S (s)$, $s\in \{1,2,3\}$. For this $s$, we define $v_\mathcal B\coloneqq v_i(s+1)$ and $w_\mathcal B\coloneqq w_{\mathcal Bs}$.

\item \label{item:five}We reveal further copies of $\mathcal S$ together with edges described below, until a copy $\mathcal C$ is assigned to all three online colors. In the copy $\mathcal C$, we denote the nodes $v_i(t)$ by $w_{\mathcal Ct}$ for clarity.
The nodes $w_{\mathcal Ct}$ of the copy $\mathcal C$ are connected to the palettes $\mathcal S(5)$, $\mathcal S(6)$, $\mathcal S(7)$ by edges depending on their offline colors, in that nodes $w_{\mathcal Ct}$ with $c(w_{\mathcal Ct})=C$ are connected to the palette nodes $v_i(s)$ for all $i\in \{1,2,3\}$ and  $s\in \{5,6,7\}$ such that $\mathcal S(s)$ is a $C$ palette, see Figure \ref{fig:csketch}.
As the nodes $w_{\mathcal Ct}$ are assigned to all three online colors, a makespan of $2$ is also achieved in the online color $i\notin \{\varphi(w_\mathcal A), \varphi(w_\mathcal B)\}$. More precisely, there exists a monochromatic edge $\{v_i(s), w_{\mathcal Ct}\}$ of the online color $i$. We rename $v_\mathcal C\coloneqq v_i(s)$ and $w_\mathcal C\coloneqq w_{\mathcal Ct}$. 

\item \label{item:six} We permute the offline coloring in each of the connected components created in \eqref{item:three}, \eqref{item:four} and \eqref{item:five} again, so that $c(v_j)=\mathrm{blue}$ and $c(w_j)=\mathrm{yellow}$ for $j\in \{\mathcal A, \mathcal B, \mathcal C\}$.

\item \label{item:seven} We add a final node $v_n$ together with hyperedges $\{v_n, v_j, w_j\}$ for each $j\in \{\mathcal A, \mathcal B, \mathcal C\}$. Now, the online color $\varphi(v_n)$ that this node is assigned to admits a makespan of $3$, while extending the offline coloring by $c(v_n)=\mathrm{red}$ maintains an offline coloring with makespan $1$. 
 \end{enumerate}

  \begin{figure}[h!]
    \centering
    \begin{tikzpicture}
      \draw[rounded corners, fill=gray!10] (0, -1) rectangle (4, 2);
    \draw[draw=black, thick, fill=gray!30] (3,1.45) to [out=0, in=180] (4.45,1.45) to [bend left=40] node[above]{$e_1$}(6,0.45) to  [out=0, in=180] (6,0.45) .. controls (6,2) and (5,3) ..
    (3,1.45);

        \draw[draw=black, thick, fill=gray!30] (3,0.45) to [out=0, in=180] (4.45,0.45) to node[above]{$e_2$}(6,0.45) to  [out=0, in=180] (6,0.45) .. controls (6,0.5) and (5,1.8) ..
    (3,0.45);

    \draw[draw=black, thick, fill=gray!30] (3,-0.55) to [out=0, in=180] (4.45,-0.55) to [bend right=40] node[below]{$e_3$}(6,0.45) to  [out=0, in=180] (6,0.45) .. controls (6,-1.1) and (5,-2.1) ..
    (3,-0.55);
 
        \draw  node [myblue](v2) at (3,0.45) {$v_2$};
      \draw  node [myblue](v1) at (3,1.45) {$v_1$};
       \draw  node [myyellow](w2) at (4.45,0.45) {$w_2$};
        \draw  node [myyellow](w1) at (4.45,1.45) {$w_1$};
           \draw  node [myyellow](w3) at (4.45,-0.55) {$w_3$};
\draw  node [myblue](v3) at (3,-0.55) {$v_3$};
\draw  node [myred](n) at (6,0.45) {$n$};
       \draw  node [] at (-2.45,1.45) {$J_1$};
       \draw  node [] at (-2.45,0.45) {$J_2$};      \draw  node [] at (-2.45,-0.55) {$J_3$};

\end{tikzpicture}
    \caption{The introduction of the last node $v_n$ leads to a makespan of $3$. Only the seven relevant nodes are shown, although the entire hypergraph has both an online coloring and a feasible offline coloring.} 
    \label{fig:highlevel}
\end{figure}
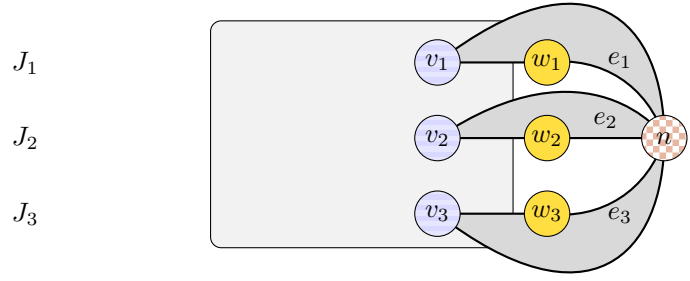

Step \eqref{item:one} is straightforward. To achieve the desired colors of $\mathcal S(1),\ldots, \mathcal S(7)$ as in step \eqref{item:two}, we find out the color $C$ of each $\mathcal S(t)$ (as a palette), and postcompose the coloring $c|_{V(\mathcal S(t))}\colon V(\mathcal S(t))\to \{\offlinecolors\}$ with the transposition \[\tau\colon \{\offlinecolors\}\to\{\offlinecolors\}\]
that swaps the color $C$ with the designated color of $\mathcal S(t)$ as stated in Step \eqref{item:two}.

\subparagraph{Step \eqref{item:three}: the subinstance $\mathcal A$.} Here, it suffices to connect the yellow palette $\mathcal S(1)$ to a new node $w_{\mathcal A}$ with the hyperedges $e=\{w_1, v_i(1)\}$ for $t\in \{1,2,3\}$. By the definition of the palette nodes $v_i(1)$, they are not yellow. Since the new node $w_{\mathcal A}$ is only connected to non-yellow nodes, we can extend the offline coloring by $c(w_{\mathcal A})=\mathrm{yellow}$, maintaining its feasibility. For the $i_{\mathcal A}$ which holds $w_{\mathcal A}\in J_{i_{\mathcal{A}}}$, we have $e=\{w_{\mathcal A}, v_{i_{\mathcal A}}(1)\}\in E(G[J_i])$ in particular. 

\subparagraph{Step \eqref{item:four}: the subinstance $\mathcal B$.}
In this stage of the construction, we use the three palettes $\mathcal S(2)$, $\mathcal S(3)$, $\mathcal S(4)$ with palette nodes $v_i(t)$ for $i\in \{1,2,3\}$ for each copy $t\in \{2,3,4\}$. Note that we are utilizing precisely one palette of each color.

Now we introduce three new nodes $w_{\mathcal B1}$,
$w_{\mathcal B2}$ and $w_{\mathcal B3}$. These nodes are to be understood as \emph{candidates} to the node $w_{\mathrm B}$ that we would like to define. They are connected to each other as well as to the three $v_i(t)$ of their respective palette, i.e., there exists the hyperedge
$\{w_{\mathcal B1},w_{\mathcal B2},w_{\mathcal B3}\}\in E$
 as well as hyperedges $\{w_{\mathcal Bt}, v_i(t+1)\}$ for $t\in \{1,2,3\}$ and $i\in \{1,2,3\}$. Accordingly, we may tentatively color the nodes as $c(w_{\mathcal B1})=\mathrm{yellow}$, $c(w_{\mathcal B2})=\mathrm{blue}$ and $c(w_{\mathcal B3})=\mathrm{red}$. 
 See Figure~\ref{fig:candidatesb} for a sketch of the introduced subhypergraph.

   \begin{figure}[h!]
    \centering
    \begin{tikzpicture}
      \draw[rounded corners, fill=yellow!30] (0, -1) rectangle (4, 2);

          \draw[rounded corners, fill=red!30] (0, -5) rectangle (4, -2);

               \draw[rounded corners, fill=blue!30] (8, -3) rectangle (12, -0);

            \coordinate  (w1) at (5,0);
      \coordinate (w2) at (5,-3);
         \coordinate(w3) at (7,-2);
       
   \draw[fill=gray!30, thick] (w1) to (w2) to(w3) to (w1);
        \draw  node [myblue](v21) at (3.5,0.45) {\tiny$v_2(2)$};
      \draw  node [myblue](v11) at (3.5,1.45) {\tiny$v_1(2)$};
       \draw  node [myblue](v22) at (3.5,-3.55) {\tiny$v_2(3)$};
        \draw  node [myyellow](v12) at (3.5,-2.55) {\tiny$v_1(3)$};
           \draw  node [myblue](v32) at (3.5,-4.55) {\tiny$v_3(3)$};
\draw  node [myblue](v31) at (3.5,-0.55) {\tiny$v_3(2)$};
\draw  node [myred](v13) at (8.5,-0.55) {\tiny$v_1(4)$};
\draw  node [myyellow](v23) at (8.5,-1.55) {\tiny$v_2(4)$};
\draw  node [myyellow](v33) at (8.5,-2.55) {\tiny$v_3(4)$};
\draw  node [](p1) at (2,0.45) {$\mathcal S(2)$};
\draw  node [](p2) at (2,-3.55) {$\mathcal S(3)$};
\draw  node [](p3) at (10,-1.55) {$\mathcal S(4)$};

   \draw  node [myyellow](w1) at (5,0) {$w_{\mathcal B1}$};
      \draw  node [myred](w2) at (5,-3) {$w_{\mathcal B2}$};
         \draw  node [myblue](w3) at (7,-2) {$w_{\mathcal B3}$};

       \draw  node [] at (-0.45,1.45) {$J_1$};
       \draw  node [] at (-0.45,0.45) {$J_2$};      
       \draw  node [] at (-0.45,-0.55) {$J_3$};
              \draw  node [] at (-0.45,-2.55) {$J_1$};
       \draw  node [] at (-0.45,-3.55) {$J_2$};      
       \draw  node [] at (-0.45,-4.55) {$J_3$};
                    \draw  node [] at (7.55,-0.55) {$J_1$};
       \draw  node [] at (7.55,-1.55) {$J_2$};      
       \draw  node [] at (7.55,-2.55) {$J_3$};
       \foreach \i in {1,...,3} {
      \foreach \j in {1,...,3}{
           \draw[draw=black, thick] (v\j\i) to (w\i);
        }
       }
\end{tikzpicture}
\caption{An example constellation of the three palettes together with the nodes $w_{\mathcal Bt}$. The $y$-axis in the palettes denote the online colorings of the palette nodes respectively. The offline colors of the palette nodes might look differently, but they are never the same color as the palette, which is depicted as the color of the rectangles. The online coloring of the nodes $w_{\mathcal Bi}$ are not yet known.}
\label{fig:candidatesb}
\end{figure}
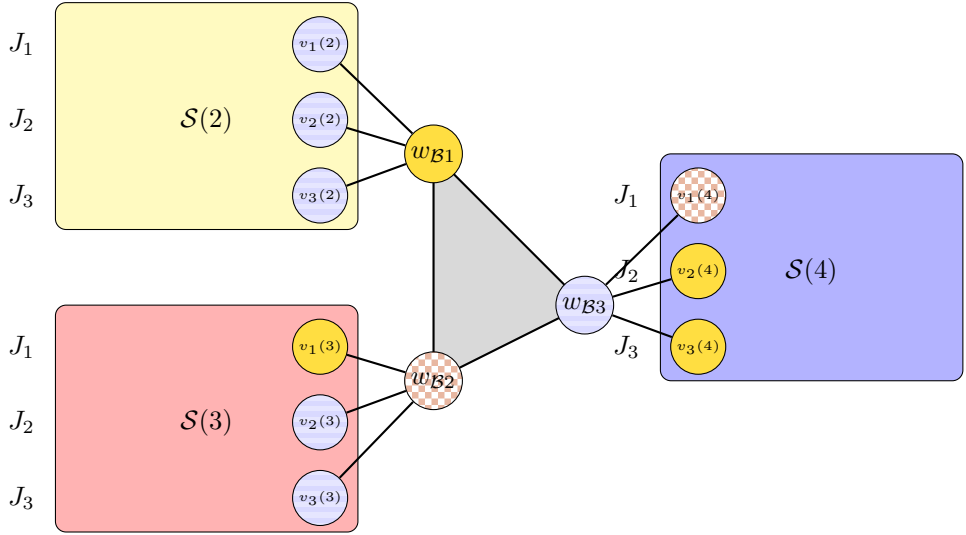

Now we would like to choose the node $w_{\mathcal B}$ ultimately. To this end, we consider any online assignment of the three candidates $w_{\mathcal Bt}$ ($t=1,2,3$). If $\varphi(w_{\mathcal B1})=\varphi(w_{\mathcal B2})=\varphi(w_{\mathcal B3})$, we obtain a makespan of $3$ and can immediately truncate the instance, as it no longer has a competitive ratio better than $3$. Hence, there are at least two online colors $i\in \{1,2,3\}$ that the candidate nodes $w_{\mathcal Bt}$ are assigned to, and at least one of them is not $i_A$. We pick such an online color and call it $i_B$. 

\subparagraph{Step \eqref{item:five}: the subinstance $\mathcal C$.}

To ensure that an edge $e_{i''}\in E(G[J_{i''}])$ is created in precisely the online color $i''\notin \{i,i'\}$, we introduce a number of candidate nodes $w_{\mathcal Cs}$ similar to the previous step. Here, it is integral that the nodes $w_{\mathcal Cs}$ are connected to each other in a way that is 
\begin{enumerate}[(i)]
    \item dense enough so that they are forced to utilize all three online colors, but also
    \item sparse enough that the offline coloring can be extended to the $w_{\mathcal Cs}$.
\end{enumerate}
Conveniently, we can recollect to our considerations about the subgraph $\mathcal S$: If we create seven copies of it, one of them is forced to attain all three online colors, or else we can complete the hypergraph in a way that no online coloring is better than $3$-competitive. 

Recalling that we may connect a node $w_j$ with the three nodes $v_i(t)$ of the palette in its offline color, we define the seven nodes as depicted in Figure \ref{fig:allseven}. More precisely, we create palettes $\mathcal S(5)$, $\mathcal S(6)$, $\mathcal S(7)$ and introduce nodes $w_{\mathcal C1},\ldots ,w_{\mathcal C7}$ together with the following hyperedges:
\begin{itemize}
    \item $\{w_{\mathcal C1}, w_{\mathcal C2},w_{\mathcal C3}\}$, $\{w_{\mathcal C2},w_{\mathcal C4}, w_{\mathcal C5}\}$, $\{w_{\mathcal C4}, w_{\mathcal C6}, w_{\mathcal C7}\}$, $\{w_{\mathcal C3}, w_{\mathcal C6}\}$, $\{w_{\mathcal C3}, w_{\mathcal C7}\}$, as well as their subsets of size $2$ (so that $G[\{w_{\mathcal C1},\ldots ,w_{\mathcal C7}\}]\cong\mathcal S$),
    \item $\{v_i(5), w_{\mathcal Cs}\}$ for $i\in \{1,2,3\}$, $s\in \{1,5,6\}$, i.e., between yellow palette nodes and new nodes colored yellow in Figure \ref{fig:allseven},
    \item $\{v_i(6), w_{\mathcal Cs}\}$ for $i\in \{1,2,3\}$, $s\in \{2,7\}$, i.e., between blue palette nodes and new nodes colored blue in Figure \ref{fig:allseven},
    \item $\{v_i(7), w_{\mathcal Cs}\}$ for $i\in \{1,2,3\}$, $s\in \{3,4\}$, i.e., between red palette nodes and new nodes colored red in Figure \ref{fig:allseven}.
\end{itemize}
An overview of the construction, without specifying the online colors of the nodes $w_{\mathcal Cs}$, is given in Figure \ref{fig:csketch}.

    \begin{figure}[h!]
    \centering
    \begin{tikzpicture}
          \coordinate (C) at (0.45,0.45);
      \coordinate(A) at (0.45,1.45);
     \coordinate(B) at (1.45,1.45);
     \coordinate(E) at (1.45,2.45);
     \coordinate(D) at (2.45,1.45);
     \coordinate(F) at (2.45,0.45);

         \coordinate(G) at (3.45,1.45);

               \draw[rounded corners, fill=yellow!30] (-4, 3) rectangle (0, 6);

          \draw[rounded corners, fill=red!30] (-4, -3) rectangle (0, 0);

               \draw[rounded corners, fill=blue!30] (4, -3) rectangle (8, 0);
\draw  node [myblue](v11) at (-0.5,5.45) {\tiny$v_1(1)$};
\draw  node [myblue](v21) at (-0.5,4.45) {\tiny$v_2(1)$};
\draw  node [myblue](v31) at (-0.5,3.45) {\tiny$v_3(1)$};

\draw  node [myyellow](v12) at (-0.5,-0.55) {\tiny$v_1(2)$};
\draw  node [myblue](v22) at (-0.5,-1.55) {\tiny$v_2(2)$};
\draw  node [myblue](v32) at (-0.5,-2.55) {\tiny$v_3(2)$};

\draw  node [myred](v13) at (4.5,-0.55) {\tiny$v_1(3)$};
\draw  node [myyellow](v23) at (4.5,-1.55) {\tiny$v_2(3)$};
\draw  node [myyellow](v33) at (4.5,-2.55) {\tiny$v_3(3)$};

         \draw[fill=gray!30, thick] (A) -- (B) -- (C) -- cycle;
\node at (barycentric cs:B=1,A=0.7,C=1) {};
  \draw[fill=gray!30, thick] (B) -- (D) -- (E) -- cycle;
\node at (barycentric cs:D=1,B=0.7,E=1) {};
  \draw[fill=gray!30, thick] (G) -- (D) -- (F) -- cycle;
\node at (barycentric cs:G=1,D=0.7,F=1) {};
    \draw[thick] (G) to [out=340, in=270, looseness=1.6] node[below]{} (C) ;
     \draw[thick] (F) to [bend left=50] node[below]{} (C) ;
     \draw  node [myred](3) at (0.45,0.45) {$w_{\mathcal C3}$};
     \draw  node [myyellow](1) at (0.45,1.45) {$w_{\mathcal C1}$};
     \draw  node [myblue](2) at (1.45,1.45) {$w_{\mathcal C2}$};
     \draw  node [myyellow](5) at (1.45,2.45) {$w_{\mathcal C5}$};
     \draw  node [myred](4) at (2.45,1.45) {$w_{\mathcal C4}$};
      \draw  node [myyellow](6) at (3.45,1.45) {$w_{\mathcal C6}$};
         \draw  node [myblue](7) at (2.45,0.45) {$w_{\mathcal C7}$};

     \foreach \j in {1,2,3}{
           \draw[draw=black, thick] (v\j1) to (1);
      }
    \foreach \j in {1,2,3}{
           \draw[draw=black, thick] (v\j1) to (5);
      }
      \foreach \j in {1,2,3}{
           \draw[draw=black, thick] (v\j1) to [out=0, in=100, looseness=1](6);
      }
          \foreach \j in {1,2,3}{
           \draw[draw=black, thick] (v\j2) to (3);
      }
                \foreach \j in {1,2,3}{
           \draw[draw=black, thick] (v\j2) to [out=0, in=250, looseness=0.2](4);
      }
                \foreach \j in {1,2,3}{
           \draw[draw=black, thick] (v\j3) to [out=180, in=270, looseness=1.3](2);
      }
                    \foreach \j in {1,2,3}{
           \draw[draw=black, thick] (v\j3) to(7);
      }

      \draw  node [](p1) at (-2,4.45) {$\mathcal S(5)$};
\draw  node [](p2) at (-2,-1.55) {$\mathcal S(7)$};
\draw  node [](p3) at (6,-1.55) {$\mathcal S(6)$};

       \draw  node [] at (-3.45,5.45) {$J_1$};
       \draw  node [] at (-3.45,4.45) {$J_2$};      
       \draw  node [] at (-3.45,3.45) {$J_3$};
              \draw  node [] at (-3.45,-0.55) {$J_1$};
       \draw  node [] at (-3.45,-1.55) {$J_2$};      
       \draw  node [] at (-3.45,-2.55) {$J_3$};
                    \draw  node [] at (7.55,-0.55) {$J_1$};
       \draw  node [] at (7.55,-1.55) {$J_2$};      
       \draw  node [] at (7.55,-2.55) {$J_3$};
\end{tikzpicture}
    \caption{Partial sketch of subinstance $\mathcal C$.}
    \label{fig:csketch}
\end{figure}
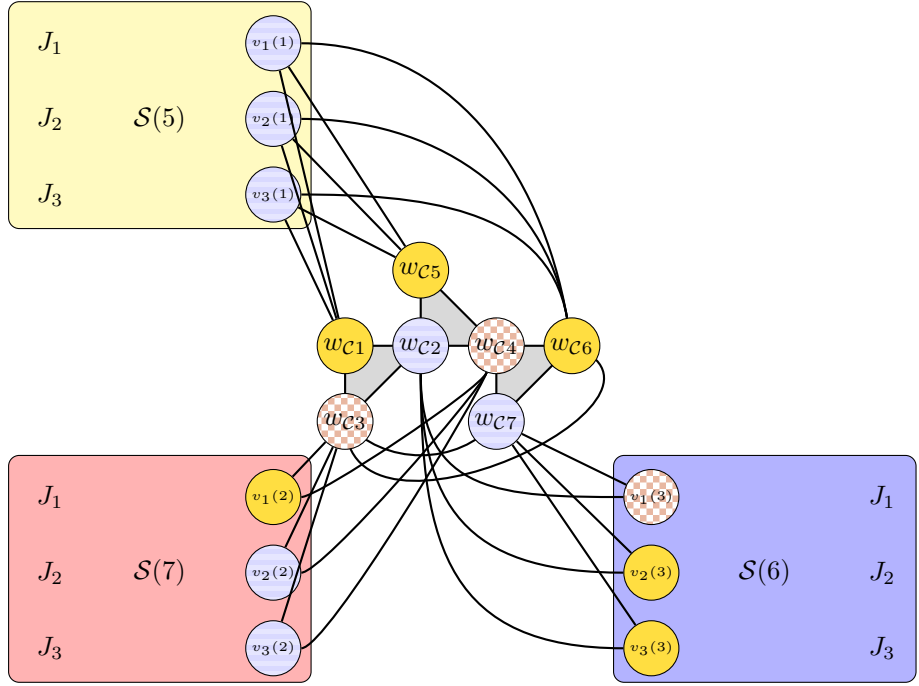

Recall that our goal is to find an edge of the form $\{v_{i''}(t), w_{\mathcal Cs}\}\subseteq J_{i''}$ for the $i''\notin \{i,i'\}$. If the nodes $w_{\mathcal Cs}$ for $s\in \{1,\ldots, 7\}$ are spread across the three online colors, there is an index $s\in \{1,\ldots, 7\}$ with the property that $w_{\mathcal Cs}\in J_{i''}$; for a suitable $t$ depending on the offline color of the node $w_{\mathcal Cs}$, we thus find the node $v_{i''}(t)$ and hence the aforementioned edge as desired.

In order to guarantee the existence of $w_{\mathcal Cs}\in J_{i''}$, we need to convince ourselves that the subgraph induced by $w_{\mathcal C1},\ldots, w_{\mathcal C7}$ is partitioned on all three online colors. Conveniently, we have seen in Subsection \ref{subs:pigeon7} that if we create seven copies of the subgraph $G[\{w_{\mathcal C1},\ldots ,w_{\mathcal C7}\}]$, one of the copies must admit all three online colors or else the assignment cannot be better than $3$-competitive by Lemma \ref{lemma:twocolor}. Therefore, we create seven copies of the subgraph induced by $w_{\mathcal C1},\ldots, w_{\mathcal C7}$, combining the nodes in each copy with the corresponding palette nodes.

In total, one of the seven copies of $w_{\mathcal C1},\ldots, w_{\mathcal C7}$ must admit all three online colors, in particular the online color $i''$. For a node $w_{\mathcal Cs}$ attaining this online color, we find a palette node $v_{i''}(t)$ such that $\{v_{i''}(t), w_{\mathcal Cs}\}\in J_{i''}$. Now we compose the offline coloring $c$ obtained so far with a permutation of the three offline colors such that $c(v_{i''}(t))=\mathrm{blue}$ and $c(w_{\mathcal Cs})=\mathrm{yellow}$.

The reader might wonder why we need all three components; indeed, we will have three hyperedges of the type $e_i$ already in the underlying hypergraph of $\mathcal C$ alone. However, we require that the hyperedges $e_i$ are each colored in blue and yellow endpoints, which is not possible when we take all of them from $\mathcal C$.

\subparagraph{Step \eqref{item:six}: Recoloring the connected components.}
Finally, we permute the colors in the three connected components created in Steps \eqref{item:three}, \eqref{item:four} and \eqref{item:five} (as shown in Figure \ref{fig:candidatesb}) such that $c(w_{\mathcal B{i_{\mathcal B}}})=\mathrm{yellow}$ and $c(v_{i_{\mathcal B}}(i_{\mathcal B}))=\mathrm{blue}$. 
See Figure \ref{fig:candidatesbperm} for an example of the permutation process.

   \begin{figure}[h!]
    \centering
    \begin{tikzpicture}

               \draw[rounded corners, fill=yellow!30] (8, -3) rectangle (12, -0);

\draw  node [myred](v13) at (8.5,-0.55) {\tiny$v_1(4)$};
\draw  node [myblue](v23) at (8.5,-1.55) {\tiny$v_2(4)$};
\draw  node [myblue](v33) at (8.5,-2.55) {\tiny$v_3(4)$};

\draw  node [](p3) at (10,-1.55) {$\mathcal S(4)$};

      \draw  node [myyellow](w3) at (7,-2.55) {$w_{\mathcal B3}$};

                    \draw  node [] at (7.55,-0.55) {$1$};
       \draw  node [] at (7.55,-1.55) {$2$};      
       \draw  node [] at (7.55,-2.55) {$3$};
       \foreach \i in {1,...,3} {
           \draw[draw=black, thick] (v\i3) to (w3);
        }
\end{tikzpicture}
\caption{Permuting the colors of the palette $\mathcal S(4)$. This permutation was implied by the assignment $\varphi(w_{\mathcal B3})=i_{\mathcal B}=3$. We have to ensure that $c(v_3(3))=\mathrm{blue}$ and $c(w_{\mathcal B3})=\mathrm{yellow}$, which is achieved by transposing the offline colors yellow and blue. Of course, this transposition is applied to the entire component, such that e.g. $c(w_{\mathcal B1})=\mathrm{blue}$, which is omitted here.}
\label{fig:candidatesbperm}
\end{figure}
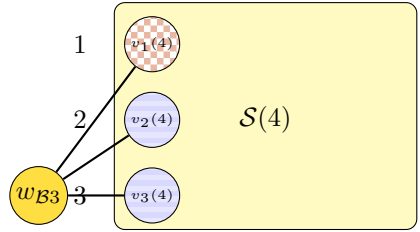 
We are now ready to execute Step \eqref{item:seven} and force instance to admit a makespan of $3$. To this end, we summarize our setting so far.

\begin{lemma}
    There exists a family $\mathcal I$ of instances over the node set $V$ together with feasible colorings $c\colon V\to \{\mathrm{blue}, \mathrm{red}, \mathrm{yellow}\}$ such that for every online assignment $\varphi\colon V\to \{1,2,3\}$, at least one instance $I\in \mathcal I$ satisfies one of the following properties:
    \begin{enumerate}[(i)]
        \item $\mathrm{makespan}(I)=3$, or
        \item there exist neighboring nodes $v_i(t),w_i\in J_i$ for each $i\in \{1,2,3\}$ such that $c(v_i(t))=\mathrm{blue}$ and $c(w_i)=\mathrm{yellow}$.
    \end{enumerate}    
\end{lemma}
\subparagraph{Step \eqref{item:seven}: the last red node}
As sketched in the introduction of this subsection (cf.~Figure \ref{fig:highlevel} in particular), we insert the very last node $n$ by combining it with the nodes $v_i(t)$ and $w_i$. More precisely, this last node is incident to the hyperedges $\{v_i(t), w_i, n\}$ for $i\in \{1,2,3\}$. No matter which online color it is assigned to, a monochromatic hyperedge of that particular online color emerges.

Notice that we have maintained a feasible offline coloring of the nodes so far up to the very last node, and the last node $n$ neighbors only blue and yellow nodes by construction and therefore the offline assignment $c(n)=\mathrm{red}$ extends the feasible coloring. In total, this yields that the offline optimal value is $1$, implying Theorem \ref{thm:nobetterthan3}. A careful enumeration of introduced hyperedges reveals that $233$ hyperedges suffice for any of the cases demonstrated, which we present next.

\subsection{Analysis}
The motivation for the aforementioned family of instances is that, already for a finite number $K$ of hyperedges, the competitive ratio is no better than the trivial bound of $3$. Although it is not necessarily the case that our construction has the minimal possible size, we would still like to analyze it so as to provide an upper bound to the smallest number of scenarios for which no algorithm with a non-trivial competitive ratio exists. To recollect our construction, we create seven copies of the subgraph $\mathcal S$ and apply a case distinction:

\begin{enumerate}[(A)]
    \item If every copy is partitioned among at most two online colors as in Subsection \ref{subs:pigeon7}, we insert three additional nodes $x_1$, $x_2$, $x_3$ and seven new hyperedges to obtain an adversarial instance. Together with the seven copies of $\mathcal S$, which have $7$ nodes and $14$ hyperedges each, this instance adds up to $7\cdot 7+3=52$ nodes and $7\cdot 14+7=105$ hyperedges.

    \item If the above case does not hold, the seven steps given previously are executed in order to create an adversarial instance. To be more precise:
    \begin{itemize}
        \item  Throughout the execution of the seven steps, at most $6$ of the copies are assumed to be partitioned onto exactly two online colors, which means up to $6\cdot 7=42$ nodes and $6\cdot 14=84$ hyperedges. 
        \item In addition to those copies, there are $7$ copies that are needed for Step \eqref{item:one}, i.e., $7\cdot 7=49$ nodes and $7\cdot 14=98$ hyperedges. 
        \item In Step \eqref{item:three}, we introduce one more node $w_\mathcal A$ with $3$ new hyperedges. 
        \item In Step \eqref{item:four}, we introduce $3$ new nodes $w_{\mathcal Bi}$ $(i\in \{1,2,3\})$. Each new node is connected to a palette by $3$ new hyperedges of size $2$ and to each other with one hyperedge of size $3$, meaning that Step \eqref{item:four} adds $3\cdot3+1=10$ new hyperedges. 
        \item Step \eqref{item:five} essentially adds an eighth copy of $\mathcal S$, i.e., $7$ nodes and $14$ hyperedges, and connects this copy to the existing instance using three hyperedges of size $2$ per inserted node, which adds up to $7\cdot3=21$ more hyperedges; i.e., $14+21=35$ for the entire step.
        \item Finally, there is the last node that leads to the competitive ratio of $3$, it is connected to the existing instance by $3$ hyperedges.
    \end{itemize}

\end{enumerate}
    In total, we have introduced up to $42+49+1+3+7+1=103$ nodes and up to $84+98+3+10+35+3=233$ hyperedges, as claimed in Theorem \ref{thm:nobetterthan3}.

\newpage
\section{Tight Bounds for Arbitrary $m$ and Sufficiently Many Scenarios}\label{sec:crazysec}

Our main goal in this section is to prove the following theorem.

\begin{theorem}\label{km_fullvers}
    Let $m\in \mathbb N_{\geq 1}$. There exists a number $N(m)$ such that there is no $(m-\varepsilon)$-competitive algorithm for any $\varepsilon>0$ for Online Makespan Hypergraph Coloring, even when restricted to hyperforests with $N(m)\leq (2m+1)\uparrow\uparrow 5$ (i.e., a tower of height $5$) nodes of dimension at most $m$.
\end{theorem}

In terms of scheduling, this translates to the following statement.

\begin{corollary}\label{thm:schedcor}
    Let $m\in \mathbb N_{\geq 1}$. There exists a number $K_m$ such that there is no $(m-\varepsilon)$-competitive algorithm for any $\varepsilon>0$, for Online Makespan Scheduling under $K$ Scenarios for $K\geq K_m$, even when restricted to unit processing times and $|S_k|\leq m$ for every scenario $S_k\in \{S_1,\ldots, S_{K_m}\}$.
\end{corollary}

One might wonder whether the number $K_m$ can be strengthened to a global $K$ for all possible numbers $m$ of machines. However, the generalization of Graham's List Scheduling Algorithm is $\max\{m,K+1\}$-competitive, even on varying processing times, answering this question negatively.

The reader might further notice the subtlety that Theorem \ref{thm:general} and Corollary \ref{thm:schedcor} are given in terms scenarios, while Corollary \ref{ackermann} and Theorem \ref{km_fullvers} are given in terms of nodes (which normally correspond to jobs). This is a deliberate choice; indeed, while our problem on scheduling parametrizes the makespan over the number $K$ of scenarios instead of the number $n$ of jobs, it is more common in hypergraph coloring literature to express the number $m$ of colors in terms of the number $n$ of nodes instead of the number $K$ of hyperedges. 

Moreover, as it is easier to count nodes in our construction, we do so throughout this section and converse our bounds to the number of hyperedges at the very end.

In an attempt to create a badly-behaving family of instances, let us turn the problem around: Now, we become the adversary and reveal for every new node $j\in [n]$ some hyperedges $e\subseteq [j]$ with the following properties:
\begin{enumerate}[(i)]
    \item $j\in e$,
    \item $e\setminus \{j\}$ has been revealed as a hyperedge before, unless $e=\{j\}$.
\end{enumerate}
The algorithm then assigns the new node to some online color, which we have no control over. 

As the adversary, we are once again interested in maintaining a feasible offline coloring in addition to the online coloring determined by the algorithm. Since we define our instance recursively, this is a much more complicated task. Therefore, we would like to maintain an \emph{active subhypergraph} $\mathcal G$ of our hypergraph $\mathcal H$ such that every connected component of $\mathcal G$ consists of exactly one maximal hyperedge, and these hyperedges are pairwise disconnected in the hypergraph $\mathcal H$. Accordingly, we also call the hyperedges $e\in E(\mathcal G)$ and their nodes $v\in e$ active, as well, as long as they are part of the active subhypergraph. 

Active hypergraphs serve a crucial purpose in our construction. If we reveal a new node $j\in [n]$ incident to multiple maximal hyperedges, at most one of these hyperedges can remain active, since they are no longer disconnected in $\mathcal H$. Throughout our construction, we always connect new nodes the existing hypergraph via extending active hyperedges. This way, we ensure that the entire hypergraph $\mathcal H$ remains a hyperforest, hence a feasible offline coloring with $d\leq m$ colors exists where $d$ is the maximum number of nodes of a hyperedge.

Inspired by the case distinction of hyperedges being split among two or three online colors in the previous construction, we describe a milestone which we intend to repeat to force larger monochromatic subgraphs. 

\begin{lemma}\label{lemma:xmd}
For \textsc{Online Makespan Hypergraph Coloring}$(m)$, there exists a family of instances $\mathcal I(m,d)$ with at most $X(m,d)$ nodes for a suitable number $X(m,d)$ each, on which the adversary achieves against any algorithm an online coloring $\varphi\colon [n]\to [m]$ with at least one of the following properties:
    \begin{enumerate}[(i)]
        \item For every online color $i\in [m]$, there is a node of color $i$, i.e., it holds that $\varphi^{-1}(\{i\})\neq \emptyset$. Moreover, there is at least one node of each online color that is active.
        \item For an online color $i\in [m]$, there exists an active monochromatic hyperedge of size $d$, i.e., $e\in E(\mathcal H)$, $|e|=d$ and $\varphi(j)=i$ for every $j\in e$.
    \end{enumerate}

    The number $X(m,d)$ is implicitly defined by the recursion 
    \begin{align*} 
    X(m,d)=\begin{cases}
        1,& d=1\text{ or }m=1\\X(m-1,d)\cdot((m-1)\cdot X(m,d-1)+1)+X(m,d-1), &\text{else}
    \end{cases}  
    \end{align*}
\end{lemma}
\begin{proof}
    We prove this statement by induction. Evidently for $m=1$, $X(1,d)=1$ node suffices to achieve the first property. Similarly for $d=1$, $X(m,1)=1$ node suffices to achieve the second property. 

    Now, we fix $(m,d)$ and assume that the statement holds for all $1\leq m'\leq m$, $1\leq d'\leq d$, $(m',d')\neq (m,d)$. First, we observe by induction that we know an instance $\mathcal I(m,d-1)$ with $X(m,d-1)$ nodes that either has a node of every online color, or a monochromatic $(d-1)$-dimensional hyperedge. If the former property holds, then the statement follows for $(m,d)$ as well. Therefore, we may assume without loss of generality that, the second property holds every time we reveal a copy of this instance (in the same order of nodes as revealed in the inductive assumption). 
    
    We execute $X(m-1,d)\cdot (m-1)+1$ copies of the instance $\mathcal I(m,d-1)$. Here, the number $X(m-1,d)$ is finite due to the inductive assumption. From each instance, we maintain solely the $(d-1)$-dimensional hyperedge (which exists by the second property of the lemma) as active. Each such hyperedge is assigned an online color. If there exists a hyperedge of every online color, then the first property holds and the statement is shown. Therefore, there are hyperedges of at most $(m-1)$ colors. By the pigeonhole principle, there is an online color (the online color $m$, without loss of generality) such that $\varphi^{-1}(\{m\})$ contains at least $X(m-1,d)+1$ pairwise node-disjoint hyperedges of size $d-1$.

    We have now secured $X(m-1,d)+1$ active hyperedges $e_1,\ldots, e_{X(m-1,d)+1}$ which we use to simulate \textsc{OMHC}$(m-1)$. To this end, we reveal a copy of the instance $\mathcal{I}(m-1,d)$. In addition to the hyperedges within $\mathcal I(m-1,d)$, we reveal for the $j$-th node the hyperedge $e_j\cup\{j\}$. If $j$ is assigned the online color $m$, the second property holds and our proof is completed. Therefore, we may assume that all nodes of the copy of $\mathcal I(m-1,d)$ are assigned one of the online colors $\{1,\ldots, m-1\}$. After every such assignment, exactly one hyperedge $e_j$ is deactivated, so that the $(d-1)$-dimensional hyperedge $e_{X(m-1,d)+1}$, which is monochromatic of color $m$, remains active. 
    
    Now we observe: Since the subhypergraph induced by the last $X(m-1,d)$ nodes is isomorphic to $\mathcal I(m-1,d)$ and the nodes must be partitioned among $m-1$ colors, one of the two properties hold; either there is a monochromatic $d$-dimensional hyperedge, or all online colors $1,\ldots, m-1$ are attained (as is $m$, due to the hyperedges). In either of the cases, the property follows for the entire instance, which we define as $\mathcal{I}(m,d)$.

    We create at most $X(m-1,d)\cdot (m-1)+1$ copies of at most $X(m,d-1)$ nodes each, followed by at most another $X(m-1,d)$ nodes. In total, this accounts to at most
  \begin{align*}
     & (X(m-1,d)\cdot (m-1)+1)\cdot X(m,d-1)+X(m-1,d)\\=&X(m-1,d)\cdot((m-1)\cdot X(m,d-1)+1)+X(m,d-1)
  \end{align*}
    nodes, as claimed. 
\end{proof}

Finding a closed form of the above recursion turns out to be challenging, though we provide a rather rough estimate for the curious reader's intuition as well as for later use.

\begin{lemma}\label{lemma:xsize}
    For $X(m,d)$ as in Lemma \ref{lemma:xmd}, we have $X(m,d)\leq \prod_{i=1}^{m+d}(m+d)^{2^i}$.
\end{lemma}
\begin{proof}
    We prove this by induction. For $m=1$ and $d=1$, the claim is vacuously true. Assuming that it holds for $1\leq d'\leq d$, $(m',d')\neq (m,d)$, we compute 
    \begin{align*}
        X(m,d)&=X(m-1,d)\cdot((m-1)\cdot X(m,d-1)+1)+X(m,d-1)\\
        &\leq (m+1)\cdot X(m-1,d)\cdot X(m,d-1)\\
        &\leq (m+1)\cdot \prod_{i=1}^{m+d-1}(m+d)^{2^i}\cdot \prod_{i=1}^{m+d-1}(m+d)^{2^i}\\
        &\leq (m+1)\cdot \prod_{i=1}^{m+d+1-i}\left((m+d)^{2^i}\right)^2\\
      &= (m+1)\cdot \prod_{i=1}^{m+d-1}(m+d)^{2^{i+1}}\\
            &\leq (m+1)\cdot \prod_{i=1}^{m+d}(m+d)^{2^{i}}
    \end{align*}
\end{proof}

 We can further bound in Lemma \ref{lemma:xsize} as 
 \[X(m,d)\leq  (m+1)\cdot \prod_{i=1}^{m+d}(m+d)^{2^{i}}\leq (m+1)(m+d)^{2^{m+d+1}}.\]
 In particular, 
 \begin{equation}\label{lastminutebound}
     X(m,m)\leq (m+1)\cdot (2m)^{2^{2m+1}}.
 \end{equation}

Both of the properties we force to exist in Lemma \ref{lemma:xmd} prove useful for our adversarial purposes. In particular, if we let $m=d$, the second property implies that the algorithm is no better than $m$-competitive, which we want to show. However, we cannot exclude the first property, which we also want to lead to an incompetitive demise in the sequel.

Having been granted a hypergraph $\mathcal H$ with a monochromatic active node of each online color, we proceed similarly as in the case of $m=3$, in that the many copies of $\mathcal H$ act as ``palettes''. 

However, now we must also \emph{grow the palettes}, i.e., the size of the palettes must increase throughout our construction. This is especially challenging because we can only extend active hyperedges (or in general, at most one hyperedge per connected component).


We try to increase monochromatic hyperedges in size by revealing sufficiently many copies of $\mathcal I(m,m)$.
\begin{lemma}\label{lemma:ymd}
For \textsc{OMHC}$(m)$, there exists a family of instances $\mathcal L(m,d)$ which are hyperforests with at most $Y(m,d)$ nodes each for a suitable number $Y(m,d)$, such that the adversary achieves on $\mathcal L(m,d)$ at least one of the following properties against any algorithm resulting in an online coloring $\varphi\colon [n]\to [m]$:
    \begin{enumerate}[(i)]
        \item For every online color $i\in [m]$, there are $d$-dimensional hyperedges $e_1,\ldots, e_m$ whose nodes are all active and the nodes in $e_i$ have the online color $i$. Furthermore, these hyperedges are pairwise disconnected in $\mathcal{L}(m,d)$.
        \item For an online color $i\in [m]$, there exists an active monochromatic hyperedge of size $m$, i.e., $e\in E(\mathcal L(m,d))$, $|e|=m$ and $\varphi(j)=i$ for every $j\in e$.
    \end{enumerate}

    The number $Y(m,d)$ is implicitly defined by the recursion 
    \begin{align*} 
    \hspace{-0.3cm}
    Y(m,d)=\begin{cases}
        1,&m=1\\
        X(m,m)&d=1\\m^2\cdot  (Y(m-1,d)^{m}+1)\cdot(Y(m,d-1)+1)+ m\cdot Y(m-1,d)^{m+1}+1, &\text{else}.
    \end{cases}  
    \end{align*}
\end{lemma}
\begin{proof}
      As in the previous proof, we prove the statement by induction. For $d=1$, we can set $\mathcal L(m,1)=\mathcal I(m,m)$ and $Y(m,1)=X(m,m)$. For $m=1$, $Y(1,d)=1$ nodes suffices to satisfy the second property. 

      Let us assume that the supposed instances $\mathcal L(m',d')$ with $Y(m',d')$ nodes each exist for every $m'\leq m, d'\leq d, (m',d')\neq (m,d)$.

      We first release 
         \begin{equation}
        a_1\coloneqq m^2\cdot Y(m-1,d)^{m}+1 
      \end{equation} 
      disjoint copies of $\mathcal L(m,d-1)$. If any of the copies satisfy the second property, the statement follows immediately. Therefore, we may assume that each copy satisfies the first property, i.e., we obtain $a_1$ hyperedges of each color. By virtue of the first property, the hyperedges within a copy are pairwise disconnected, so that we can mark all of them as active. 

      Next, we release $a_1$ singleton nodes. These extend $m$ hyperedges per node; one in each color. Therefore, they extend precisely one monochromatic hyperedge of size $d-1$ to a monochromatic hyperedge of size $d$ per node. We mark the new monochromatic hyperedges as active and all others as inactive.
      
     Unfortunately, the online algorithm assigns the singleton nodes, which we have no control over. However, we can guarantee that in one of the colors, in the $m$-th color without loss of generality, there are at least 
     \begin{equation}
         a_2\coloneqq m\cdot Y(m-1,d)^{m}+1
     \end{equation}
     active hyperedges $e_1,\ldots, e_{a_2}$. 

    Then, we release $m\cdot Y(m-1,d)^{m-1}$ copies of $\mathcal L(m-1,d)$. We connect these copies to our existing active hyperedges as follows: In addition to the hyperedges within $\mathcal L(m-1,d)$, each node in the copy extends one of the hyperedges $e_k$, e.g., one with the smallest size that is active, excluding $e_{a_2}$. This way, every node extends one active hyperedge of size $(d-1)$ of colors $1,\ldots, m-1$ each.
    
    Let us call a copy of $\mathcal L(m-1,d)$ \emph{successful} if the algorithm assigns no node $v$ in this copy the color $m$, and \emph{unsuccessful} otherwise. 
    
    \begin{claim}\label{claim:unsuccessful}
        There can be at most $m\cdot Y(m-1,d)^m-1$ unsuccessful copies unless the second property is satisfied.
    \end{claim} 
    \begin{proof}[Proof of Claim \ref{claim:unsuccessful}]
     For the sake of contradiction, let us assume that there are $m\cdot Y(m-1,d)^m$ unsuccessful copies. Consider the first $m\cdot Y(m-1,d)^m$ of these copies. For every unsuccessful copy, at least one active hyperedge of size $d$ must be extended to size $d+1$ and hence may stay active, and in return, $Y(m-1,d)-1$ hyperedges must become inactive. Since we always pick the smallest-size hyperedges to extend, the first $Y(m-1,d)^{m-1}$ copies exhaust all hyperedges $e_k$ (except $e_{a_2}$ which we ignore until the very end of our construction). In the process, $Y(m-1,d)^{m-1}$ hyperedges stay active, which are used once again for the next $m\cdot Y(m-1,d)^m$ copies and so on.
     
     After the $d'$-th run of this procedure for $d'\leq m-d$, we have $Y(m-1,d)^{d+m-d'}$ active hyperedges of size $d+d'$. At the end, for $d'=m-d$, we have one hyperedge of size $m$, satisfying the second property.  
    \end{proof}
    
    By the claim above, at least one copy must be successful, i.e., partitioned among colors $1,\ldots, m-1$. Once we find a successful copy, the inductive assumption on $\mathcal L(m-1,d)$ delivers one $d$-dimensional hyperedge of colors $1,\ldots, m-1$ each, which are pairwise disconnected. Together with $e_{a_2}$, this yields the first property overall.

    It remains to count the numbers of nodes and hyperedges.

\begin{enumerate}[(i)]
    \item    We first added $a_1$ copies of $\mathcal L(m,d-1)$, which accounts to $a_1\cdot Y(m,d-1)$ nodes. 
    \item Then, we added $a_1$ singleton nodes.
    \item Lastly, we released $a_2-1$ copies of $\mathcal L(m-1,d)$ together with a very last node, which accounts to $a_2\cdot Y(m-1,d)+1$ nodes. 
\end{enumerate}

In total, we utilized
\begin{align}
    &a_1\cdot(Y(m,d-1)+1)+a_2\cdot Y(m-1,d)+1\\=&m^2\cdot  (Y(m-1,d)^{m}+1)\cdot(Y(m,d-1)+1)+ m\cdot Y(m-1,d)^{m+1}+1
\end{align}
nodes, as claimed.    
\end{proof}

It is worth noting that the terms in the above proof are not tight: For the sake of simplicity, we have rounded up some terms without changing the asymptotics of the result. The previous lemma almost immediately implies Theorem \ref{km_fullvers}, which we ultimately wanted to prove.

\begin{proof}[Proof of Theorem \ref{km_fullvers}]
   We choose $N(m)\coloneqq Y(m,m-1)+1$, where $Y\colon \mathbb N\times \mathbb N\to \mathbb N$ is defined as in Lemma \ref{lemma:ymd}. Indeed, we could build the instance $\mathcal L(m,m-1)$, which provides pairwise disconnected hyperedges $e_1,\ldots, e_m$ of each color $1,\ldots, m$ by Lemma \ref{lemma:ymd}. Now a last node $v$ can be introduced so that it extends $e_1,\ldots, e_m$. For the color $i$ that the node $v$ is assigned, $e_i\cup\{v\}$ is a monochromatic edge of color $i$.

   By the recursion in Lemma \ref{lemma:ymd}, it follows that 
\[Y(m,d)\leq (Y(m-1,d)\cdot Y(m,d-1))^{m+2}\leq \max\{Y(m-1,d)\cdot Y(m,d-1)\}^{2m+4}.\]
Extending the terms $Y(m,d-1)$ and $Y(m-1,d)$ and iterating this process, we obtain
\[Y(m,m-1)\leq X(m,m)^{(m+2)^{m-1}}\leq \left((m+1)\cdot {(2m)^{2^{(2m+1)}}}\right)^{(m+2)^{m-1}}\]
by \eqref{lastminutebound}.  

Very roughly, we can estimate $N(m)\leq (2m+1)\uparrow\uparrow 5$ (i.e., a tower of height $5$), as claimed in Theorem \ref{km_fullvers}.

\end{proof}

While we are still parameterizing over nodes, let us first deduce Corollary \ref{ackermann}. Theorem \ref{km_fullvers} reveals that we cannot color hyperforests with with $(2m+1)\uparrow\uparrow 5$ nodes in $m$ colors without creating a monochromatic hyperedge. In other words, if we want a feasible coloring of a hyperforest with $2n+1$ nodes, we need at least $\mathrm{slog_5(n)}$ colors, as claimed.

Throughout this section, we have counted nodes instead of hyperedges as it is stated in Theorem \ref{thm:general}. However, hyperforests with hyperedges of size at most $m$ with $n$ nodes have $O(mn)$ hyperedges, as $mn$ is the number of incidences and there cannot be more hyperedges than incidences. As both Theorem \ref{thm:general} and Corollary \ref{thm:schedcor} only state the finiteness of the number $K_m$ of scenarios, they follow immediately from Theorem \ref{km_fullvers} as well.

\section{Tight Bounds For $K=3$ and Unit Processing Times}
In the previous two sections, we have established in two ways that for large enough $K$, there is no better-than-$3$-competitive algorithm for \textsc{OMSS$(3,K)$}, even for unit processing times. Let us reverse the parameters and ask: How competitively can we solve \textsc{OMSS$(m,3)$}? 
This is not an easy question, indeed, even for \textsc{OMSS$(m,1)$}, a tight bound is not known for most $m$. Restricting to the unit processing time case $(p_j=1)$, we achieve a tight bound, which is the smallest possible ratio considering the lower bound of $\rho\geq 2$ provided by Theorem \ref{thm:lb2K}. We present a round robin algorithm that is $2$-competitive in the special case of $p_j\equiv 1$ for an arbitrary number $m$ of machines.

\begin{theorem}\label{thm:bingo}
For all $m\in \mathbb N$, there is a $2$-competitive algorithm for the special case of \textsc{OMSS$(m,3)$} where the jobs have unit processing times.
\end{theorem}

The algorithm is best visualized through an infinite sequence of ``bingo cards'' with $m$ rows and $3$ columns, each of which looks like Figure \ref{fig:bingocard}.
\begin{figure}[h!]
    \centering
    \begin{tikzpicture}
    \foreach \j in {1,...,8} {
        \foreach \i in {1,...,3} {
            \draw [very thick] (\i+1,\j) -- (\i,\j+1) node [above right] at (\i+0.5,\j+0.5) {\i};
        }
            \draw [very thick] (1,\j+1) -- (2,\j)  node [below left] at (1.5,\j+0.5) {$23$};
            \draw [very thick] (2,\j+1) -- (3,\j)  node [below left] at (2.5,\j+0.5) {$13$};
            \draw [very thick] (3,\j+1) -- (4,\j)  node [below left] at (3.5,\j+0.5) {$12$};
            \draw  node [right] at (0,9-\j+0.5) {\j};
    }
\draw [step=1.0, very thick] (1,1) grid (4,9);
\end{tikzpicture}
    \caption{Bingo card for $m=8$. Although it admittedly rather resembles a tombola card turned sideways, our analysis will be better suited for the game of bingo.}
    \label{fig:bingocard}
\end{figure}

Here, every row represents a machine, the above example illustrates $m=8$. Each of the triangles correspond to a single job, where jobs $j\in S_1\cap S_2\cap S_3$ are to be understood as an entire square. 

We assign jobs only by looking at the column they belong to, where we fill the columns in a round-robin procedure. For the first column, the counting index of the round robin starts in the first row, for the second, it starts in the $\left(\left\lceil\frac{m}{3}\right\rceil+1\right)$-st row, and for the third, it starts in the $\left(\left\lceil\frac{2m}{3}\right\rceil+1\right)$-st. We count downwards until the lowest row of the column is reached, then we proceed from the top downwards. See Figure \ref{fig:countingindex} for an example.

\begin{figure}[h!]
    \centering
    \begin{tikzpicture}
    \foreach \j in {1,...,8} {
            \draw  node [right] at (0,9-\j+0.5) {\j};
             \draw  node [right] at (1.25,9-\j+0.5) {\j};
    }
    \foreach \j in {1,...,5} {
             \draw  node [right] at (2.25,6-\j+0.5) {\j};
    }
      \foreach \j in {6,7,8} {
             \draw  node [right] at (2.25,14-\j+0.5) {\j};
    }
     \foreach \j in {1,...,2} {
             \draw  node [right] at (3.25,3-\j+0.5) {\j};
    }
         \foreach \j in {3,...,8} {
             \draw  node [right] at (3.25,11-\j+0.5) {\j};
    }
\draw [step=1.0, very thick] (1,1) grid (4,9);
   \draw[very thick, <->] (4.5,1) -- (4.5,3) node [right] at (4.5,2){$m-\left\lceil\frac{2m}{3}\right\rceil$ rows};
       \draw[very thick, <->] (4.5,3) -- (4.5,6) node [right] at (4.5,4.5){$\left\lceil\frac{2m}{3}\right\rceil-\left\lceil\frac{m}{3}\right\rceil$ rows};
       \draw[very thick, <->] (4.5,6) -- (4.5,9) node [right] at (4.5,7.5){$\left\lceil\frac{m}{3}\right\rceil$ rows};
\end{tikzpicture}
    \caption{Counting indices of the above bingo card.}
    \label{fig:countingindex}
\end{figure}

We place each job $j\in [n]\setminus(S_1\cap S_2\cap S_3)$ to the first triangle it fits into according to the counting index described above. If the job $j$ does not fit into any triangle, we create a new bingo card. 

We place a job $j\in S_1\cap S_2\cap S_3$ in an empty square with largest possible counting index. We break ties in favor of the rightmost column. If such a square does not exist, we create a new bingo card.

This procedure is formalized in Algorithm \ref{alg:bingo} for the special case that all jobs fit a single bingo card. The general procedure is that, we find the earliest-created bingo card, and apply the lines 9-19 of Algorithm \ref{alg:bingo} on it. Here, the counting indices must then be maintained for every bingo card separately.

We will first analyze instances where \ref{alg:bingo} does not assign two jobs onto the same placeholder (e.g., two jobs $j\in S_1\cup S_2\setminus S_3$ to the same triangle or a job $j\in S_1\cup S_2\cap S_3$ on top of a non-vacant triangle). At the end, we argue how we can reduce the general case to this special case.

Let us first elaborate on how Algorithm \ref{alg:bingo} exactly represents our informal description: The variables of the form $i_S$, $S\subsetneq \{1,2,3\}$ represent running indices for the triangles of \emph{type} $S$, i.e., placeholders for jobs $j\in S_k \iff k\in S$. They show the smallest-counting indexed vacant triangle. Similarly, the variables of the form $i_{\{1,2,3\}}^k$ for $k\in \{1,2,3\}$ are the backwards running indices for squares in column $k$. They show the largest-indexed vacant square (i.e., set of two corresponding triangles both of which are vacant) on column $k$. Lines 13 and 18 ensure that the indices are never out of bounds. Moreover, Line 15 assigns a job $j\in S_1\cap S_2\cap S_3$ to a highest-counting indexed square, breaking ties in favor of a largest column index, as desired.
\begin{algorithm}
\caption{The bingo card algorithm for $m$ arbitrary, $K=3$, $p_j\equiv1$, assuming that jobs fit in a bingo card.}
\label{alg:bingo}
\begin{algorithmic}[1]
\State $J_1,J_2,J_3\gets \emptyset$
\State  $i_{\{1\}}, i_{\{2,3\}} \gets 1$  
\State  $i_{\{2\}}, i_{\{1,3\}} \gets \left\lceil\frac{m}{3}\right\rceil+1$  
\State  $i_{\{3\}}, i_{\{1,2\}} \gets \left\lceil\frac{2m}{3}\right\rceil$  
\State  $i_{\{1,2,3\}}^1 \gets m$  
\State  $i_{\{1,2,3\}}^2 \gets \left\lceil\frac{m}{3}\right\rceil$ 
\State  $i_{\{1,2,3\}}^3 \gets \left\lceil\frac{2m}{3}\right\rceil$ 
\For{$j=1$ \textbf{to} $n$}
\State Let $S\subseteq [3]$ be the set of scenario indices such that $j\in S_k$
\If{$S\neq [3]$}
\State $J_{i_{S}}\gets J_{i_{S}}\cup \{j\}$
\State $i_{S}\gets i_{S}+1$
\State  $i_{S}\gets (i_{S} \mod m)$
\Else
\State $k\gets \mathrm{argmax}_{k\in [3]}\{i_{\{1,2,3\}}^k+\nicefrac{k}{4}\}$
\State $J_{i_{\{1,2,3\}}^k }\gets J_{i_{\{1,2,3\}}^k }\cup \{j\}$
\State $i_{\{1,2,3\}}^k \gets i_{\{1,2,3\}}^k -1$
\State $i_{\{1,2,3\}}^k\gets 
(i_{\{1,2,3\}}^k \mod m)$
\EndIf
\EndFor
\end{algorithmic}  
\end{algorithm}

\subparagraph{Analysis.} We are ready to analyze Algorithm \ref{alg:bingo} by defining our first type of bingo. Recall that the instance that we analyze consists of a single bingo card.
\begin{defn}
    A \emph{row bingo}, given a solution $J_1\dot\cup J_2\dot\cup J_3$ is a machine $i$, for which a scenario $k\in \{1,2,3\}$ with $p(J_i\cap S_k)=3$ exists.
\end{defn}

The following is an immediate consequence of this definition:
\begin{prop}
    Any schedule without a row bingo is $2$-competitive.
\end{prop}
\begin{proof}
    Any nonempty schedule without a row bingo has makespan at most $2$, and an offline optimum of at least $1$.
\end{proof}

Therefore, we may assume without loss of generality that the $m'$-th machine is a row bingo for some $m'=\left\lceil\frac{m_1\cdot m}{3}\right\rceil+m_2$ ($m_1\in \{0,1,2\}$, $m_2\in \{0,\ldots, \left\lceil\frac{m}{3}\right\rceil-1\}$ or $m_1=1, m_2=\left\lceil\frac{m}{3}\right\rceil, m\equiv 2\mod{3}$ or $m_1=2, m_2=\left\lceil\frac{m}{3}\right\rceil, m\equiv 0\mod{3}$).

Our goal is to show that $p(S_k)>m$, which implies that the offline optimum in the $k$-th scenario is at least $2$, bounding the competitive ratio by $\frac{3}{2}<2$. We apply a case distinction. To this end, we first observe zones into which a bingo card is decomposed by the round robin counting indices. For an example, see Figure \ref{fig:bingocardareas}.

\begin{figure}[h!]
    \centering
    \begin{tikzpicture}
        \draw [fill=yellow, opacity=0.5]
       (1,8) -- (4,8) -- (4,9) -- (1,9);
        \draw [fill=yellow, opacity=0.5]
       (2,1) -- (2,3) -- (4,3) -- (4,1);
       \draw [fill=yellow, opacity=0.5]
       (3,3) -- (2,3) -- (2,6) -- (3,6);   
        \draw [fill=blue, opacity=0.5]
       (4,3) -- (3,3) -- (3,4) -- (4,4);
       \draw[very thick, <->] (4.5,1) -- (4.5,3) node [right] at (4.5,2){$m-\left\lceil\frac{2m}{3}\right\rceil$ rows};
       \draw[very thick, <->] (4.5,3) -- (4.5,6) node [right] at (4.5,4.5){$\left\lceil\frac{2m}{3}\right\rceil-\left\lceil\frac{m}{3}\right\rceil$ rows};
       \draw[very thick, <->] (4.5,6) -- (4.5,9) node [right] at (4.5,7.5){$\left\lceil\frac{m}{3}\right\rceil$ rows};
    \foreach \j in {1,...,8} {
        \foreach \i in {1,...,3} {
            \draw [very thick] (\i+1,\j) -- (\i,\j+1) node [above right] at (\i+0.5,\j+0.5) {\i};
        }
            \draw [very thick] (1,\j+1) -- (2,\j)  node [below left] at (1.5,\j+0.5) {$23$};
            \draw [very thick] (2,\j+1) -- (3,\j)  node [below left] at (2.5,\j+0.5) {$13$};
            \draw [very thick] (3,\j+1) -- (4,\j)  node [below left] at (3.5,\j+0.5) {$12$};
            \draw  node [right] at (0,9-\j+0.5) {\j};
    }
\draw [step=1.0, very thick] (1,1) grid (4,9);
\end{tikzpicture}
    \caption{The same bingo card. The yellow colored squares were filled by at least one triangle, while the blue square has been filled by a square (i.e., a job $j\in S_1\cap S_2\cap S_3$). Notice that by the translated round-robin indexing, zones of nearly equal heights occur. If the first row is a row bingo for, say, the first scenario, the yellow areas must contain jobs $j\in S_1$.}
    \label{fig:bingocardareas}
\end{figure}
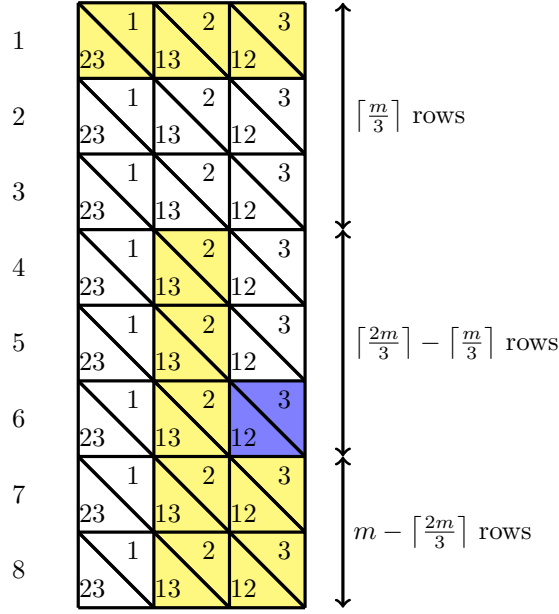

The near-equal partitioning of the rows depend on the value $m\!\!\mod 3$, which might be at times surprisingly aggravating to calculate with, which is why we have provided the reader Table \ref{tab:heights} as an aid.

\def\arraystretch{1.5}
\begin{table}[]
    \centering
  \begin{tabular}{ | c | c | c | c |} 
  \hline
   &$\left\lceil\frac{m}{3}\right\rceil$  &  $\left\lceil\frac{2m}{3}\right\rceil-\left\lceil\frac{m}{3}\right\rceil$& $m-\left\lceil\frac{2m}{3}\right\rceil$\\ 
  \hline
  0 & $\frac{m}{3}$ & $\frac{m}{3}$  &$\frac{m}{3}$\\ 
  \hline
  1 & $\frac{m+2}{3}$& $\frac{m-1}{3}$&$\frac{m-1}{3}$\\ 
  \hline
  2 & $\frac{m+1}{3}$&$\frac{m+1}{3}$&$\frac{m-2}{3}$\\
  \hline
\end{tabular}
    \caption{Numbers of rows in each partition, depending on $m\!\!\mod 3$.}
    \label{tab:heights}
\end{table}

\noindent\textbf{Case 1: There is no job $j\in S_1\cap S_2\cap S_3$ in the $m'$-th machine.} Let $k\in \{1,2,3\}$ be the scenario in the $m'$-th machine attaining the makespan of $3$. We use Table \ref{tab:heights} to find a lower bound. To be more precise, depending on $m\!\!\mod 3$, numbers from one of the three rows in this table are used to bound $S_k$ from below; it is at least two times one entry plus another entry (plus $3$ which comes from the $m'$-th row). Therefore we have 
\[p(S_k)\geq \min\left\{\frac{2m}{3}+\frac{m}{3}+3, \frac{2m-2}{3}+\frac{m-1}{3}+3, \frac{2m-4}{3}+\frac{m+1}{3}+3\right\}>m.\]

For the next cases, we observe the following:

\begin{obs}
    Fix a column $k$. Then in this column, the counting index of  a job $j\in S_1\cap S_2\cap S_3$ is always strictly larger than that of a job $j'\notin S_1\cap S_2\cap S_3$.
\end{obs}
\textbf{Case 2: There are at least two jobs $j,j'\in S_1\cap S_2\cap S_3$ in the $m'$-th machine.} 

If there are at least two jobs $j,j'\in J_{m'}\cap S_1\cap S_2\cap S_3$, then one of them must have counting index larger than $\frac{m-2}{3}$ by Table \ref{tab:heights}, therefore all three columns have at least $\frac{m-2}{3}$ jobs $j\in S_1\cap S_2\cap S_3$ not counting the $m'$-th row. In total, it must hold that 
\[S_k\geq 3\cdot\frac{m-2}{3}+3>m.\]
\textbf{Case 3: There is exactly one job $j\in S_1\cap S_2\cap S_3$ in the $m'$-th machine.}
If the job $j$ has counting index $s$ larger than $\frac{m-2}{3}$ (from bottom to top), then the case goes analogous to Case 2. Otherwise, we branch into two subcases:

\textit{Subcase 3.1: $j$ is not in the leftmost column.} Then the other two jobs not being in $S_1\cap S_2\cap S_3$ implies that there are at least 
\[\min\left\{
\left\lceil\frac{m}{3}\right\rceil+2\left(\left\lceil\frac{2m}{3}\right\rceil-\left\lceil\frac{m}{3}\right\rceil-s\right), 
m-\left\lceil\frac{2m}{3}\right\rceil+2\left(\left\lceil\frac{m}{3}\right\rceil-s\right)
\right\}\geq m-2s,\]
jobs that precede the $m'$-th row in the counting order in the other two columns; the inequality can be seen via Table \ref{tab:heights}. Together with the three jobs in $J_{m'}$ and the $3(s-1)$ jobs that precede $j$ in the backwards counting order, we obtain \[p(S_k)\geq m-2s+3+3s-3>m,\]as desired.

\textit{Subcase 3.2: $j$ is in the leftmost column.} Then the job $j'\in J_{m'}$ in the second column is preceded by  
\[\left\lceil\frac{2m}{3}\right\rceil-\left\lceil\frac{m}{3}\right\rceil+\left(m-\left\lceil\frac{2m}{3}\right\rceil-s\right)=m-\left\lceil\frac{m}{3}\right\rceil-s\] jobs in the second column, and the job $j''\in J_{m'}$ in the third column is preceded by $m-\left\lceil\frac{2m}{3}\right\rceil-s$ jobs in the third column, all belonging to $S_k$. Moreover, there are at least $3s-1$ jobs in $S_1\cap S_2\cap S_3$ preceding $j$ in the backwards counting order. Here, we used the fact that tiebreaks were broken in favor of the rightmost slots. Together with the three jobs in $J_{m'}$, we can bound
\[p(S_k)\geq 2m-\left\lceil\frac{2m}{3}\right\rceil-\left\lceil\frac{m}{3}\right\rceil-2s+3s+3>m.\]

\subparagraph{Multiple bingo cards.} Now we consider the general case where we have occupied multiple bingo cards. Recall that in the algorithm, the job is placed to the lowest-index bingo card where a vacant placeholder exists.

We call a column (consisting of triangles) of a bingo card \emph{column bingo} if there exists a scenario that is represented in every entry of this column. Then, by the minimality of the card index that we select, there can be only one non-empty and non-bingo copy of each of the three columns. 

{The bottleneck is achieved when there are no bingo columns because if we drop a bingo column, the makespan decreases by at most $1$ (exactly $1$ if and only if a scenario that achieves the bingo also achieves the makespan), while average load ratio output divided by average load becomes strictly larger. Here, notice that in our analysis, we always constructed a lower bound using the scenario that attains makespan. Therefore, there is a worst-case instance with only one card.}

\newpage
\section{Conclusion and Outlook}

Throughout, we have established several competitiveness results for Online Makespan Scheduling under Scenarios. The first main takeaway from our work is a competitiveness gap between two and at least three scenarios when we consider $m=2$ machines. This result draws a parallel to the known tractability results of several other problems, in which the line between easy and hard was drawn between two and three scenarios. Surely, the fact that the subsets of $\{1,2,3\}$ are not laminar, has been a deciding factor in our simple non-competitiveness result, as has been in several of the well-known results. It is worth mentioning that restricting to the proxy competitive ratio as well as fixing the method of assigning double-scenario jobs essentially reduced the existence problem to the description of a polyhedron. We believe that this technique, if understood better, might prove useful for other cases and related problems as well. 

We have further compared the behavior of competitiveness for increasing number $K$ of scenarios and increasing number $m$ of machines. In the special setting of unit processing times and $m=3$ machines, the contrast is evident: On one hand, there is a $2$-competitive algorithm for three scenarios and $m$ machines for all $m\in \mathbb N$, which matches the lower bound obtained readily for $m=2$. On the other hand, for $K$ sufficiently large, the trivial upper bound of $3$ cannot be beaten. That being said, $K+1$ is a strict upper bound as well for a problem on $K$ scenarios. To summarize, for large values of $m$ and small values of $K$, we would expect the possibilities for competitive ratios to be dominated by $K$.

In light of our results, we may raise several open questions: In addition to the obvious task of finding tight bounds for all possibilities $(m,K)$ of the numbers of machines and scenarios, exploring the competitiveness as $m$ approaches to infinity in the weighted processing time case remains interesting. However, these problems have not been resolved even for $K=1$, whence it is fair to assume challenges in further progress. Another interesting research direction is to introduce \emph{migration}, as seen in e.g.~\cite{migration}. In this model, we are allowed to revoke a bounded number of decisions in hindsight. Similarly, instead of minimizing the worst-case scenario, one could look into minimizing the average scenario or maximum regret.

\bibliography{onlinemakespan}
\end{document}

%% file: checkerboard.tex
\newlength{\flexcheckerboardsize}

\newcommand{\defineflexcheckerboard}[4]{
    \setlength{\flexcheckerboardsize}{#2}
    \pgfdeclarepatterninherentlycolored{#1}
        {\pgfpointorigin}{\pgfqpoint{2\flexcheckerboardsize}    
        {2\flexcheckerboardsize}}
        {\pgfqpoint{2\flexcheckerboardsize}
        {2\flexcheckerboardsize}}%
        {
            \pgfsetfillcolor{#4}
            \pgfpathrectangle{\pgfpointorigin}{
            \pgfqpoint{2.1\flexcheckerboardsize}    
                {2.1\flexcheckerboardsize}}
          \pgfusepath{fill}
          \pgfsetfillcolor{#3}
          \pgfpathrectangle{\pgfpointorigin}
            {\pgfqpoint{\flexcheckerboardsize}
            {\flexcheckerboardsize}}
          \pgfpathrectangle{\pgfqpoint{\flexcheckerboardsize}
            {\flexcheckerboardsize}}
            {\pgfqpoint{\flexcheckerboardsize}
            {\flexcheckerboardsize}}
            \pgfusepath{fill}
        }
}

\defineflexcheckerboard{flexcheckerboard_bw}{.5mm}{black}{white}
\defineflexcheckerboard{flexcheckerboard_redblue}{.5mm}{red}{blue}
\defineflexcheckerboard{flexcheckerboard_greenorange}{1mm}{green}{orange}
\defineflexcheckerboard{mycheck}{0.8mm}{white}{BrickRed!30}